\documentclass[10pt,letterpaper,twoside,reqno]{amsart}

\makeatletter{}
\usepackage{fixltx2e}                       
\usepackage[usenames,dvipsnames]{xcolor}
\usepackage{fancyhdr}
\usepackage{amsmath,amsfonts,amsbsy,amsgen,amscd,mathrsfs,amssymb,amsthm}
\usepackage{subfig}
\usepackage{url}

\usepackage{mathtools}

\usepackage[font=small,margin=0.25in,labelfont={sc},labelsep={colon}]{caption}

\usepackage{tikz}
\usepackage{microtype}
\usepackage{enumitem}

\definecolor{dark-gray}{gray}{0.3}
\definecolor{dkgray}{rgb}{.4,.4,.4}
\definecolor{dkblue}{rgb}{0,0,.5}
\definecolor{medblue}{rgb}{0,0,.75}
\definecolor{rust}{rgb}{0.5,0.1,0.1}

\usepackage[colorlinks=true]{hyperref}

\hypersetup{urlcolor=rust}
\hypersetup{citecolor=dkblue}
 \hypersetup{linkcolor=dkblue}

\usepackage{setspace}

\usepackage{graphicx}
\usepackage{booktabs,longtable,tabu} 
\usepackage{multirow} 
\usepackage{float}

\usepackage[T1]{fontenc}

\usepackage{fourier}
\usepackage{charter}

\usepackage{bm} 
\graphicspath{{figures/}}

\newtheorem{theorem}{Theorem}[section]

\newtheorem{proposition}[theorem]{Proposition}
\newtheorem{fact}[theorem]{Fact}

\newtheorem{corollary}[theorem]{Corollary}

\theoremstyle{definition}

\newtheorem{definition}[theorem]{Definition}
\newtheorem{example}[theorem]{Example}
\newtheorem{remark}[theorem]{Remark}

\newcommand{\term}{\emph}

\numberwithin{equation}{section} 
\numberwithin{figure}{section}
\numberwithin{table}{section}

\floatstyle{plaintop}
\newfloat{recipe}{thp}{lor}
\floatname{recipe}{Recipe}
\numberwithin{recipe}{section}

\providecommand{\mathbold}[1]{\bm{#1}}

\renewcommand{\phi}{\varphi}

\newcommand{\eps}{\varepsilon}

\newcommand{\half}{\tfrac{1}{2}}

\newcommand{\econst}{\mathrm{e}}

\newcommand{\Id}{\mathbf{I}}

\newcommand{\sphere}[1]{\mathsf{S}^{#1}}
\newcommand{\ball}[1]{\mathsf{B}^{#1}}

\providecommand{\mathbbm}{\mathbb} 
\newcommand{\R}{\mathbbm{R}}

\newcommand{\polar}{\circ}

\newcommand{\abs}[1]{\left\vert {#1} \right\vert}
\newcommand{\abssq}[1]{{\abs{#1}}^2}

\newcommand{\diff}[1]{\mathrm{d}{#1}}
\newcommand{\idiff}[1]{\, \diff{#1}}

\newcommand{\Prob}{\mathbbm{P}}

\newcommand{\Expect}{\operatorname{\mathbb{E}}}

\newcommand{\normal}{\textsc{normal}}

\DeclareMathOperator{\Var}{Var}

\newcommand{\vct}[1]{\mathbold{#1}}
\newcommand{\mtx}[1]{\mathbold{#1}}

\newcommand{\transp}{\mathsf{t}}
\newcommand{\adj}{*}

\newcommand{\nullity}{\operatorname{null}}

\newcommand{\trace}{\operatorname{trace}}

\newcommand{\psdge}{\succcurlyeq}

\newcommand{\ip}[2]{\left\langle {#2},\ {#1} \right\rangle}
\newcommand{\absip}[2]{\abs{\ip{#1}{#2}}}
\newcommand{\abssqip}[2]{\abssq{\ip{#1}{#2}}}

\newcommand{\norm}[1]{\left\Vert {#1} \right\Vert}
\newcommand{\normsq}[1]{\norm{#1}^2}

\DeclareMathOperator{\dist}{dist}

\newcommand{\smnorm}[2]{{\bigl\Vert {#2} \bigr\Vert}_{#1}}

\newcommand{\enorm}[1]{\norm{#1}}
\newcommand{\enormsm}[1]{\enorm{\smash{#1}}}

\newcommand{\fnorm}[1]{\norm{#1}_{\mathrm{F}}}
\newcommand{\fnormsq}[1]{\fnorm{#1}^2}

\newcommand{\pnorm}[2]{\norm{#2}_{#1}}

\newcommand{\Desc}{\mathcal{D}}

\newcommand{\cone}{\operatorname{cone}}

\newcommand{\minimize}[1]{\underset{#1}{\text{minimize}}\quad}
\newcommand{\subjto}{\quad\text{subject to}\quad}

\title[The bowling scheme]{Convex recovery of a structured signal \\ from independent random linear measurements}

\author[J.~A.~Tropp]{Joel~A.~Tropp}

\hypersetup{pdftitle={The Bowling Scheme}}
\hypersetup{pdfauthor={Joel A.~Tropp}}

\date{5 May 2014. Revised 9 September 2014 and 2 December 2014.}

\begin{document}

\begin{abstract}
This chapter develops a theoretical analysis of the convex programming method
for recovering a structured signal from independent random linear measurements.
This technique delivers bounds for the sampling complexity that are
similar with recent results for standard Gaussian measurements, but
the argument applies to a much wider class of measurement ensembles.
To demonstrate the power of this approach, the paper presents
a short analysis of phase retrieval by trace-norm minimization.
The key technical tool is a framework, due to Mendelson and coauthors,
for bounding a nonnegative empirical process.
\end{abstract}

\maketitle

\section{Motivation}

Signal reconstruction from random measurements is a central preoccupation in contemporary signal processing.  In this problem, we acquire linear measurements of an unknown, structured signal through a random sampling process.
Given these random measurements, a standard method for recovering the unknown signal is to solve a convex optimization problem that enforces our prior knowledge about the structure.  The basic question is how many measurements suffice to resolve a particular type of structure.

Recent research has led to a comprehensive answer when the measurement operator follows the standard Gaussian distribution~\cite{MPT07:Reconstruction-Subgaussian,RV08:Sparse-Reconstruction,stojnic10,OH:10,CRPW12:Convex-Geometry,ALMT14:Living-Edge-II,FoyMac:13,OymHas:13,OTH13:Simple-Bounds,TOH14:Simple-Error}.
The literature also contains satisfying answers for subgaussian measurements~\cite{MPT07:Reconstruction-Subgaussian} and subexponential measurements~\cite{Men10:Empirical-Processes}.  Other types of measurement systems are quite common, but we are not aware of a simple approach that allows us to analyze general measurements in a unified way.

This chapter describes an approach that addresses a wide class of convex signal reconstruction problems involving random sampling.  To understand these questions, the core challenge is to produce a lower bound on a nonnegative empirical process.  For this purpose, we rely on a powerful framework, called the \term{Small Ball Method}, that was developed by Shahar Mendelson and coauthors in a sequence of papers, including~\cite{KM13:Bounding-Smallest,Men13:Remark-Diameter,Men14:Learning-Concentration,LM14:Compressed-Sensing,Men14:Learning-Concentration-II}.

To complete the estimates required by Mendelson's Small Ball Method, we propose a technique based on conic duality.  One advantage of this approach is that we can exploit the same insights and calculations that have served so well in the Gaussian setting.  We refer to this little argument as the \emph{bowling scheme} in honor of David Gross's \emph{golfing scheme}~\cite{Gro11:Recovering-Low-Rank}.  We anticipate that it will offer researchers an effective way to analyze many signal recovery problems with random measurements.

\subsection{Roadmap}

The first half of the chapter summarizes the established analysis of convex signal reconstruction with a Gaussian sampling model.  In Section~\ref{sec:convex-recovery}, we introduce a convex optimization framework for solving structured signal recovery problems with linear measurements, and we present a geometric formulation of the optimality conditions.  Section~\ref{sec:gauss-model} specializes to the case where the measurements come from a Gaussian model, and we explain how classical results for Gaussian processes lead to a sharp bound for the number of Gaussian measurements that suffice.  These results are framed in terms of a geometric parameter, the conic Gaussian width, associated with the convex optimization problem.  Section~\ref{sec:width-descent} explains how to use duality to obtain a numerically sharp bound for the conic Gaussian width, and it develops two important examples in detail.

In the second half of the chapter, we consider more general sampling models.
Section~\ref{sec:beyond-gauss} introduces Mendelson's Small Ball Method and
the technical arguments that support it.
As a first application, in Section~\ref{sec:subgauss-model}, we use this strategy to analyze signal reconstruction from subgaussian measurements.  Section~\ref{sec:bowling} presents the bowling scheme, which merges the conic duality estimates with Mendelson's Small Ball Method.  This technique allows us to study more general types of random measurements.
Finally, in Section~\ref{sec:phase}, we demonstrate the vigor of these ideas by applying them to the phase retrieval problem.

\section{Signal reconstruction from linear measurements}
\label{sec:convex-recovery}

We begin with a framework that describes many convex optimization methods for recovering a structured signal from linear measurements.  Examples include the $\ell_1$ minimization approach for identifying a sparse vector and the Schatten 1-norm minimization approach for identifying a low-rank matrix.  We develop a simple error bound for convex signal reconstruction by exploiting the geometric formulation of the optimality conditions.  This analysis leads us to study the minimum conic singular value of a matrix.

\subsection{Linear acquisition of data}

Let $\vct{x}^\natural \in \R^d$ be an unknown but ``structured'' signal.  Suppose that we observe a vector $\vct{y}$ in $\R^m$ that consists of $m$ linear measurements of the unknown:
\begin{equation} \label{eqn:data-acquisition}
\vct{y} = \mtx{\Phi} \vct{x}^{\natural} + \vct{e}.
\end{equation}
We assume that $\mtx{\Phi}$ is a known $m \times d$ sampling matrix, and $\vct{e} \in \R^m$ is a vector of unknown errors.  The expression~\eqref{eqn:data-acquisition} offers a model for data acquisition that describes a wide range of problems in signal processing, statistics, and machine learning.  Our goal is to compute an approximation of the unknown $\vct{x}^\natural$ by exploiting our prior knowledge about its structure.

\subsection{Reconstruction via convex optimization}

Convex optimization is a popular approach for recovering a structured vector from linear measurements.  Let $f : \R^d \to \overline{\R}$ be a proper convex function\footnote{The extended real numbers $\overline{\R} := \R \cup \{ \pm \infty \}$.  A \term{proper} convex function takes at least one finite value but never the value $-\infty$.}
that reflects the ``complexity'' of a signal.  Then we can frame the convex program
\begin{equation} \label{eqn:convex-recovery}
\minimize{\vct{x} \in \R^d}	f(\vct{x})
\subjto		\enormsm{ \mtx{\Phi} \vct{x} - \vct{y} } \leq \eta
\end{equation}
where $\enorm{\cdot}$ denotes the Euclidean norm and $\eta$ is a specified bound on the norm of the error $\vct{e}$.  In words, the optimization problem~\eqref{eqn:convex-recovery} searches for the most structured signal $\vct{x}$ that is consistent with the observed data $\vct{y}$.  In practice, it is common to consider the Lagrangian formulation of~\eqref{eqn:convex-recovery} or to consider a problem where the objective and constraint are interchanged.  We can often solve~\eqref{eqn:convex-recovery} and its variants efficiently using standard algorithms.

\begin{remark}[Alternative programs]
The optimization problem~\eqref{eqn:convex-recovery} is not the only type of convex method
for signal reconstruction.  Suppose that $f : \R^d \to \overline{\R}$ is a gauge, i.e.,
a function that is nonnegative, positively homogeneous, and convex.  Then we may consider the convex program
\begin{equation*} \minimize{\vct{x} \in \R^d}	f(\vct{x})
\subjto		f^{\polar}\big(\mtx{\Phi}^\transp( \mtx{\Phi} \vct{x} - \vct{y}) \big) \leq \eta,
\end{equation*}
where $f^\polar$ denotes the polar of the gauge~\cite[Chap.~15]{Roc70:Convex-Analysis}
and ${}^\transp$ denotes transposition.
This reconstruction method submits to an analysis similar with
the approach in this note.  For example, see~\cite[Thm.~1]{TLR14:Geometrizing-Local}.
\end{remark}

\subsection{Examples}

Before we continue, let us mention a few structures that arise in applications and the complexity measures that are typically associated with these structures.

\begin{example}[Sparse vectors]
A vector $\vct{x}^\natural \in \R^d$ is \term{sparse} when many or most of its entries are equal to zero.  We can promote sparsity by minimizing the $\ell_1$ norm $\pnorm{\ell_1}{\cdot}$.  This heuristic leads to a problem of the form
\begin{equation} \label{eqn:l1-min}
\minimize{\vct{x} \in \R^d}	\pnorm{\ell_1}{\vct{x}}
\subjto		\enormsm{ \mtx{\Phi} \vct{x} - \vct{y} } \leq \eta.
\end{equation}
Sparsity has become a dominant modeling tool in statistics, machine learning, and signal processing.   
\end{example}

\begin{example}[Low-rank matrices]
We say that a matrix $\mtx{X}^\natural \in \R^{d_1 \times d_2}$ has \term{low rank} when its rank is small compared with minimum of $d_1$ and $d_2$.  Suppose that we have acquired noisy measurements
\begin{equation} \label{eqn:matrix-acquisition}
\vct{y} = \mtx{\Phi}(\mtx{X}^\natural) + \vct{e},
\end{equation}
where $\mtx{\Phi}$ is a linear operator that maps a matrix in $\R^{d_1 \times d_2}$ to a vector in $\R^m$.  To reconstruct the unknown low-rank matrix $\mtx{X}^\natural$, we can minimize the Schatten 1-norm $\pnorm{S_1}{\cdot}$, which returns the sum of the singular values of a matrix.  This heuristic suggests that we consider an optimization problem of the form
\begin{equation} \label{eqn:S1-min}
\minimize{\vct{X} \in \R^{d_1 \times d_2}}	\pnorm{S_1}{\mtx{X}}
\subjto		\enormsm{ \mtx{\Phi}( \vct{X} ) - \vct{y} } \leq \eta.
\end{equation}
In recent years, this approach to fitting low-rank matrices has become common.
\end{example}

\noindent
It is possible to consider many other types of
structure.  For instance, see~\cite{CRPW12:Convex-Geometry,FoyMac:13}.

\subsection{A deterministic error bound for convex recovery}

We can obtain a deterministic error bound for the convex reconstruction method~\eqref{eqn:convex-recovery} using a standard geometric analysis. Recall that a \term{cone} is a set $K \subset \R^d$ that is positively homogeneous: $K = \tau K$ for all $\tau > 0$.  A \term{convex cone} is a cone that is also a convex set.  Let us introduce the cone of descent directions of a convex function.

\begin{definition}[Descent cone]
Let $f : \R^d \to \overline{\R}$ be a proper convex function.  The \term{descent cone} $\Desc(f, \vct{x})$ of the function $f$ at a point $\vct{x} \in \R^d$ is defined as
$$
\Desc(f, \vct{x}) := \bigcup_{\tau > 0} \big\{ \vct{u} \in \R^d : f(\vct{x} + \tau \vct{u}) \leq f(\vct{x}) \big\}.
$$
The descent cone of a convex function is always a convex cone, but it may not be closed.
\end{definition}

\noindent
We are interested in the behavior of the measurement matrix $\mtx{\Phi}$ when it is restricted to a descent cone.

\begin{definition}[Minimum conic singular value]
\label{def:conic-sing}
Let $\mtx{\Phi}$ be an $m \times d$ matrix, and let $K$ be a cone in $\R^d$.  The minimum singular value of $\mtx{\Phi}$ with respect to the cone $K$ is defined as
$$
\lambda_{\min}( \mtx{\Phi}; K ) := \inf\big\{ \enorm{ \mtx{\Phi} \vct{u} } :
	\vct{u} \in K \cap \sphere{d-1} \big\}
$$
where $\sphere{d-1}$ is the Euclidean unit sphere in $\R^d$.
\end{definition}

\noindent
The terminology originates in the fact that $\lambda_{\min}(\mtx{\Phi}; \R^{d})$ coincides with the usual minimum singular value. 

\begin{figure}[t!]
\begin{center}
\newcommand{\conelen}{2.75}
\newcommand{\myAlp}{12}
\newcommand{\myA}{2}
\begin{tikzpicture}[scale=1.75]
\fill[white!90!blue] (0,1) -- ++(225:\conelen) arc(225:315:\conelen) -- cycle;
\draw[thin,blue] (0,1) -- ++(225:\conelen) (0,1) -- ++(315:\conelen);

\fill[white!70!blue] (1,0) -- (0,1) -- (-1,0) -- (0,-1) -- cycle;
\draw[very thick] (1,0) -- (0,1) -- (-1,0) -- (0,-1) -- cycle;

\path (0,1) +(90-\myAlp:0.5) coordinate(tubeuc) {};
\path (0,1) +(270-\myAlp:0.5) coordinate(tubebc) {};

\path (tubeuc) +(180-\myAlp:\myA) coordinate(tubeul)  {};
\path (tubeuc) +(-\myAlp:\myA) coordinate(tubeur) {};
\path (tubebc) +(-\myAlp:\myA) coordinate(tubebr) {};
\path (tubebc) +(180-\myAlp:\myA) coordinate(tubebl) {};

\path (0,1) +(180-\myAlp:\myA) coordinate(tubecl) {};
\path (0,1) +(-\myAlp:\myA) coordinate(tubecr) {};

\fill[white!60!red,semitransparent] (tubeul) -- (tubeur) -- (tubebr) -- (tubebl) -- cycle;

\draw[red,thin] (tubeul) -- (tubeur);
\draw[red,thin] (tubebl) -- (tubebr);

\path (tubeur) +(-\myAlp:0.1) coordinate(tubearrowu) {};
\path (tubecr) +(-\myAlp:0.1) coordinate(tubearrowc) {};
\path (tubebr) +(-\myAlp:0.1) coordinate(tubearrowb) {};

\draw[->,>=latex,red!80!black,thick] (tubebl) +(-0.75,0) node[below] {$\norm{\mtx{\Phi} \vct{u} } \leq 2\eta$} to [out=90, in=180-\myAlp] +(0.3,0.15);

\draw[very thick,black!30!red] (tubecl) -- (tubecr);

\draw[->,>=latex,black!50!red,thick] (tubecl) +(-0.5,0.5) node[above] {$\nullity(\mtx{\Phi})$} to [out=-90, in = 180-\myAlp] +(0,0);

\draw[->,>=latex,thick,black!60!blue] (-1.25,0.25) node[left]{$\{\vct{u} : f(\vct{x}^\natural + \vct{u}) \leq f(\vct{x}^\natural) \}$} to [out=30, in=135] (-0.6,0.125);

\fill (0,1) circle (1pt) node[above right=-1pt]{$\vct{0}$};

\path (-1.2,-0.8) node[below,black!10!blue,scale=1]{$\Desc(f,\vct{x}^\natural)$};

\end{tikzpicture}

\caption{\textbf{[Geometry of convex recovery]}   This diagram illustrates the geometry of the optimization problem~\eqref{eqn:optimization-descent}.  The cone $\Desc(f,\vct{x}^\natural)$ contains the directions $\vct{u}$ in which $f$ is decreasing at $\vct{x}^\natural$.  Assuming that $\enorm{\vct{e}} \leq \eta$, the diagonal tube contains every point $\vct{u}$ that satisfies the bound constraint $\norm{\mtx{\Phi} \vct{u} + \vct{e}} \leq \eta$.  Each optimal point $\widehat{\vct{u}}$ for~\eqref{eqn:optimization-descent} lies in the intersection of the tube and the cone.} 
\label{fig:convex-recovery}
\end{center}
\end{figure}

With these definitions at hand, we reach the following basic result.

\begin{proposition}[A deterministic error bound for convex recovery] \label{prop:deterministic-error}
Let $\vct{x}^\natural$ be a signal in $\R^d$, let $\mtx{\Phi}$ be an $m \times d$ measurement matrix, and let
$\vct{y} = \mtx{\Phi} \vct{x}^\natural + \vct{e}$ be a vector of measurements in $\R^m$.
Assume that $\enorm{\vct{e}} \leq \eta$, and let $\widehat{\vct{x}}_{\eta}$ be any solution to the optimization problem~\eqref{eqn:convex-recovery}.  Then
$$
\smnorm{}{ \widehat{\vct{x}}_{\eta} - \vct{x}^\natural } \leq \frac{2 \eta}{\lambda_{\min}\big(\mtx{\Phi};\ \Desc( f, \vct{x}^\natural ) \big)}.
$$
\end{proposition}

\noindent
This statement is adapted from~\cite{CRPW12:Convex-Geometry}.
For completeness, we include the short proof.

\begin{proof}
It is natural to write the decision variable $\vct{x}$ in the convex program~\eqref{eqn:convex-recovery} relative to the true unknown: $\vct{u} := \vct{x} - \vct{x}^\natural$.  Using the expression~\eqref{eqn:data-acquisition} for the measurement vector $\vct{y}$, we obtain the equivalent problem
\begin{equation} \label{eqn:optimization-descent}
\minimize{\vct{u} \in \R^d} f(\vct{x}^\natural + \vct{u})
\subjto \enorm{ \mtx{\Phi} \vct{u} - \vct{e} } \leq \eta.
\end{equation}
Owing to the bound $\enorm{\vct{e}} \leq \eta$, the point $\vct{u} = \vct{0}$ is feasible for~\eqref{eqn:optimization-descent}.  Therefore, each optimal point $\widehat{\vct{u}}$ verifies $f(\vct{x}^\natural + \widehat{\vct{u}}) \leq f(\vct{x}^\natural)$.
In summary, any optimal point of~\eqref{eqn:optimization-descent} satisfies two conditions:
$$
\widehat{\vct{u}} \in \Desc(f, \vct{x}^\natural)
\quad\text{and}\quad
\enorm{ \mtx{\Phi} \widehat{\vct{u}} - \vct{e} } \leq \eta.
$$
As a consequence, we simply need to determine how far we can travel in a descent direction before we violate the bound constraint.    See Figure~\ref{fig:convex-recovery} for an illustration of the geometry.

To complete the argument, assume that $\vct{u}$ is a nonzero point in $\Desc(f, \vct{x}^\natural)$ that is feasible for~\eqref{eqn:optimization-descent}.  Then
$$
\lambda_{\min}\big(\mtx{\Phi}; \ \Desc(f,\vct{x}^\natural) \big)
	\leq \frac{ \enorm{\mtx{\Phi} \vct{u}} }{ \enorm{\vct{u}} }
	\leq \frac{ \enorm{\mtx{\Phi} \vct{u} - \vct{e}} + \enorm{\vct{e}}}
	{ \enorm{\vct{u}} }
	\leq \frac{2 \eta}{\enorm{\vct{u}}}.
$$
The first inequality follows from Definition~\ref{def:conic-sing} of the conic singular value.  The second relation is the triangle inequality.  The last bound holds because $\vct{u}$ satisfies the constraint in~\eqref{eqn:optimization-descent}, and we have assumed that $\enorm{\vct{e}} \leq \eta$.  Finally, rearrange the display, and rewrite $\vct{u}$ in terms of the original decision variable $\vct{x}$.
\end{proof}

Although Proposition~\ref{prop:deterministic-error} is elegant, it can be difficult to apply because we must calculate the minimum conic singular value of a matrix $\mtx{\Phi}$ with respect to a descent cone.  This challenge becomes less severe, however, when the matrix $\mtx{\Phi}$ is drawn at random.

\section{A universal error bound for Gaussian measurements}
\label{sec:gauss-model}

We will study the prospects for convex recovery when the sampling matrix $\mtx{\Phi}$ is chosen at random.  This modeling assumption arises in signal processing applications where the matrix describes a data-acquisition system that can extract random measurements.  This kind of model also appears in statistics and machine learning when each row of the matrix tabulates measured variables for an individual subject in an experiment.

\subsection{Standard Gaussian measurements}

In this section, we treat one of the simplest mathematical models for the $m \times d$ random measurement matrix $\mtx{\Phi}$.  We assume that each of the $m$ rows of $\mtx{\Phi}$ is drawn independently from the standard Gaussian distribution $\normal(\vct{0}, \Id_d)$, where the covariance $\Id_d$ is the $d$-dimensional identity matrix.  For this special case, we can obtain a sharp estimate for the minimum conic singular value $\lambda_{\min}(\mtx{\Phi}; K)$ for any convex cone $K$.

\subsection{The conic Gaussian width}

The analysis of Gaussian sampling depends on a geometric summary parameter for cones.

\begin{definition}[Conic Gaussian width] \label{def:conic-gauss-width}
Let $K \subset \R^d$ be a cone, not necessarily convex.  The \term{conic Gaussian width} $w(K)$ is defined as
$$
w(K) := \Expect{} \sup\nolimits_{\vct{u} \in K \cap \sphere{d-1}} \
	\ip{\vct{u}}{\smash{\vct{g}}}
$$
where $\vct{g} \sim \normal(\vct{0}, \Id_d)$ is a standard Gaussian vector in $\R^d$.
\end{definition}

\noindent
The Gaussian width plays a central role in asymptotic convex geometry~\cite{MS86:Asymptotic-Theory,Pis89:Volume-Convex,LT91:Probability-Banach}.  Most of the classical techniques for bounding widths are only accurate up to constant factors (or worse).  In contrast, ideas from the contemporary signal processing literature frequently allow us to produce numerically sharp estimates for the Gaussian width of a cone.  These techniques were developed in the papers~\cite{stojnic10,OH:10,CRPW12:Convex-Geometry,ALMT14:Living-Edge-II,FoyMac:13}.  We will outline one of the methods in Section~\ref{sec:width-descent}.

\begin{remark}[Statistical dimension]
The conic Gaussian width $w(K)$ is a convenient functional because it arises
from the probabilistic tools that we use.  The theory of conic integral geometry,
however, delivers a better summary parameter~\cite{ALMT14:Living-Edge-II}.
The \term{statistical dimension} $\delta(K)$ of a convex cone $K$ can be defined as
$$
\delta(K) := \Expect{} \big[ \big(\sup\nolimits_{\vct{u} \in K \cap \ball{d}}
	\ \ip{ \vct{u} }{ \smash{\vct{g}} } \big)^2 \big],
$$
where $\ball{d}$ is the Euclidean unit ball in $\R^d$ and $\vct{g} \sim \normal(\vct{0}, \Id_d)$.
The statistical dimension canonically extends the dimension
of a subspace to the class of convex cones, and it satisfies many elegant
identities~\cite[Prop.~3.1]{ALMT14:Living-Edge-II}.
For some purposes, the two parameters are interchangeable
because of the following comparison~\cite[Prop.~10.2]{ALMT14:Living-Edge-II}:
$$
w^2(K) \leq \delta(K) \leq w^2(K) + 1.
$$
As a consequence, we can interpret $w^2(K)$ as a rough measure of the ``dimension'' of a cone.
\end{remark}

\subsection{Conic singular values and conic Gaussian widths}

As it turns out, the conic Gaussian width $w(K)$ controls the minimum conic singular value $\lambda_{\min}(\mtx{\Phi}; K)$ when $\mtx{\Phi}$ follows the standard normal distribution.

\begin{proposition}[Minimum conic singular value of a Gaussian matrix] \label{prop:conic-sing-gauss}
Let $K \subset \R^d$ be a cone, not necessarily convex,
and let $\mtx{\Phi}$ be an $m \times d$ matrix whose rows are independent vectors drawn from the standard Gaussian distribution $\normal(\vct{0}, \Id_d)$.  Then
$$
\lambda_{\min}(\mtx{\Phi}; K) \geq \sqrt{m-1} - w(K) - t
$$
with probability at least $1 - \econst^{-t^2/2}$.
\end{proposition}

\noindent
In essence, this result dates to the work of Gordon~\cite{Gor85:Some-Inequalities,Gor88:Milmans-Inequality}.  We have drawn the proof from the survey~\cite[Sec.~3.2]{DS01:Local-Operator} of Davidson \& Szarek; see also~\cite{MPT07:Reconstruction-Subgaussian,RV08:Sparse-Reconstruction,stojnic10,CRPW12:Convex-Geometry}.  Note that the argument relies on special results for Gaussian processes that do not extend to other distributions.

\begin{proof}[Proof sketch]
We can express the minimum conic singular value as
$$
\lambda_{\min}(\mtx{\Phi}; K)
	= \inf_{\vct{u}\in K \cap \sphere{d-1}} \ \sup_{\vct{v} \in \sphere{m-1}} \
	\ip{ \mtx{\Phi} \vct{u} }{ \vct{v} }
$$
It is a consequence of Gordon's comparison inequality~\cite[Thm.~1.4]{Gor85:Some-Inequalities} that
$$
\Expect{} \inf_{\vct{u} \in K \cap \sphere{d-1}} \ \sup_{\vct{v} \in \sphere{m-1}}
	\ip{ \mtx{\Phi} \vct{u} }{ \vct{v} }
	\geq \Expect{} \sup_{\vct{v} \in \sphere{m-1}} \ip{ \vct{v} }{ \smash{\vct{g}'} }
	- \Expect{} \sup_{\vct{u} \in K \cap \sphere{d-1}} \ip{ \vct{u} }{ \smash{\vct{g}} }
	= \Expect{} \enormsm{\vct{g}'} - w(K),
$$
where $\vct{g}' \sim \normal(\vct{0}, \Id_m)$ and $\vct{g} \sim \normal(\vct{0},\Id_d)$.
It is well known that $\Expect{} \enormsm{\vct{g}'} \geq \sqrt{m-1}$, and therefore
\begin{equation} \label{eqn:expect-conic-sing-gauss}
\Expect \lambda_{\min}(\mtx{\Phi}; K)
	\geq \sqrt{m-1} - w(K).
\end{equation}
To complete the argument, note that the map
$$
\lambda_{\min}(\cdot; K) : \mtx{A} \mapsto \inf_{\vct{u} \in K \cap \sphere{d-1}} \norm{\mtx{A} \vct{u}}
$$
is 1-Lipschitz with respect to the Frobenius norm.  The usual Gaussian concentration inequality~\cite[Sec.~5.4]{BLM13:Concentration-Inequalities} implies that
\begin{equation} \label{eqn:tail-conic-sing-gauss}
\Prob\big\{ \lambda_{\min}(\mtx{\Phi}; K) \leq \Expect \lambda_{\min}(\mtx{\Phi}; K) - t \big\} \leq \econst^{-t^2/2}.
\end{equation}
Introduce the lower bound~\eqref{eqn:expect-conic-sing-gauss} for the expectation of the minimum conic singular value
into~\eqref{eqn:tail-conic-sing-gauss} to reach the advertised result.
\end{proof}

\begin{remark}[Sharpness for convex cones] \label{rem:sharp-convex}
It is a remarkable fact that the bound in Proposition~\ref{prop:conic-sing-gauss} is essentially sharp.
For any cone $K$, we can reinterpret the statement as saying that
$$
\lambda_{\min}(\mtx{\Phi}; K) > 0
\quad\text{with high probability when}\quad
m \geq w^2(K) + C w(K).
$$
(The letter $C$ always denotes a positive absolute constant, but its value may change from place to place.)
Conversely, for a convex cone $K$, it holds that
\begin{equation} \label{eqn:conic-sing-gauss-converse}
\lambda_{\min}(\mtx{\Phi}; K) = 0
\quad\text{with high probability when}\quad
m \leq w^2(K) - C w(K).
\end{equation}
The result~\eqref{eqn:conic-sing-gauss-converse} follows from research of Amelunxen et al.~\cite[Thm.~I and Prop.~10.2]{ALMT14:Living-Edge-II}.  This claim can also be derived by supplementing the proof of Proposition~\ref{prop:conic-sing-gauss} with a short polarity argument.  It is productive to interpret the pair of estimates in this remark as a \term{phase transition} for convex signal recovery; see~\cite{ALMT14:Living-Edge-II} for more information.
\end{remark}

\subsection{An error bound for Gaussian measurements}

Combining Proposition~\ref{prop:deterministic-error} and Proposition~\ref{prop:conic-sing-gauss}, we obtain a general error bound for convex recovery from Gaussian measurements.

\begin{corollary}[Signal recovery from Gaussian measurements] \label{cor:recovery-gauss}
Let $\vct{x}^\natural$ be a signal in $\R^d$.  Let $\mtx{\Phi}$ be an $m \times d$ matrix whose rows are independent random vectors drawn from the standard Gaussian distribution $\normal(\vct{0}, \Id_d)$, and let $\vct{y} = \mtx{\Phi} \vct{x}^\natural + \vct{e}$ be a vector of measurements in $\R^m$.  With probability at least $1 - \econst^{-t^2/2}$, the following statement holds.  Assume that $\enorm{\vct{e}} \leq \eta$, and let $\widehat{\vct{x}}_{\eta}$ be any solution to the optimization problem~\eqref{eqn:convex-recovery}.  Then
$$
\smnorm{}{ \widehat{\vct{x}}_{\eta} - \vct{x}^\natural } \leq \frac{2 \eta}{\big[ \sqrt{m-1} - w\big(\Desc(f, \vct{x}^\natural)\big) - t \big]_+}.
$$
The operation $[a]_+ := \max\{a, 0\}$ returns the positive part of a number.
\end{corollary}

\noindent
The overall argument that leads to this result was proposed by Rudelson \& Vershynin~\cite[Sec.~4]{RV08:Sparse-Reconstruction}; the statement here is adapted from~\cite{CRPW12:Convex-Geometry}.

Corollary~\ref{cor:recovery-gauss} provides for stable recovery of the unknown $\vct{x}^\natural$ when the number $m$ of measurements satisfies
$$
m \geq w^2\big( \Desc(f, \vct{x}^\natural) \big) + C w\big( \Desc(f, \vct{x}^\natural) \big).
$$
In view of Remark~\ref{rem:sharp-convex}, Corollary~\ref{cor:recovery-gauss} provides a refined estimate for the amount of information that suffices to identify a structured vector from Gaussian measurements via convex optimization.

\begin{remark}[The normal error model]
It is possible to improve the error bound in Corollary~\ref{cor:recovery-gauss}
if we instate a Gaussian model for the error vector $\vct{e}$.
See the papers~\cite{OymHas:13,OTH13:Simple-Bounds,TOH14:Simple-Error}
for an analysis of this case.
\end{remark}

\section{Controlling the width of a descent cone via polarity}
\label{sec:width-descent}

As soon as we know the conic Gaussian width of the descent cone,
Corollary~\ref{cor:recovery-subgauss} yields error bounds for convex recovery of a structured signal from Gaussian measurements.  To make use of this result, we need technology for calculating these widths.  This section describes a mechanism, based on polarity, that leads to extremely accurate estimates.  We can trace this method to the papers~\cite{stojnic10,OH:10}, where it is couched in the language of duality for cone programs.
The subsequent papers~\cite{CRPW12:Convex-Geometry,ALMT14:Living-Edge-II} rephrase these ideas in a more geometric fashion.
It can be shown that the approach in this section gives sharp results for many natural examples;
see~\cite[Thm.~4.3]{ALMT14:Living-Edge-II} or~\cite[Prop.~1]{FoyMac:13}.
Although polar bounds for widths are classic in asymptotic convex geometry~\cite{MS86:Asymptotic-Theory,Pis89:Volume-Convex,LT91:Probability-Banach}, the refined arguments here are just a few years old.

\subsection{Polarity and weak duality for cones}

We begin with some classical facts about conic geometry.  
\begin{fact}[Polarity]
Let $K$ be a general cone in $\R^d$.  The \term{polar cone} $K^\polar$ is the closed convex cone
$$
K^\polar := \big\{ \vct{v} \in \R^d : \ip{ \vct{x} }{ \vct{v} } \leq 0 \ \text{for all $\vct{x} \in K$} \big\}.
$$
It is easy to verify that $K \subset (K^{\polar})^\polar$.
\end{fact}

\noindent
Recall that the \term{distance} from a point $\vct{x} \in \R^d$ to a set $E \subset \R^d$
is defined by the relation
$$
\dist(\vct{x}, E) := \inf_{\vct{u} \in E} \ \norm{ \vct{x} - \vct{u} }.
$$
With these definitions, we reach the following weak duality result.

\begin{proposition}[Weak duality for cones] \label{prop:weak-duality}
Let $K$ be a general cone in $\R^d$.
For $\vct{x} \in \R^d$,
$$
\sup_{\vct{u} \in K \cap \sphere{d-1}} \ \ip{ \smash{\vct{u}} }{ \vct{x} }
	\leq \dist( \vct{x}, K^\polar ).
$$
\end{proposition}

\begin{proof}
The argument is based on a simple duality trick.  First, write
$$
\dist( \vct{x}, K^\polar ) = \inf_{\vct{v} \in K^\polar} \ \norm{ \vct{x} - \vct{v} }
	= \inf_{\vct{v} \in K^\polar} \ \sup_{\vct{u} \in \sphere{d-1}} \ \ip{ \smash{\vct{u}} }{ \vct{x} - \vct{v} }.
$$
Apply the inf--sup inequality:
$$
\dist( \vct{x}, K^\polar ) \geq \sup_{\vct{u} \in \sphere{d-1}} \ \inf_{\vct{v} \in K^\polar} \
	\ip{ \smash{\vct{u}} }{ \vct{x} - \vct{v} }
	=  \sup_{\vct{u} \in \sphere{d-1}} \bigg[ \ip{ \smash{\vct{u}} }{ \vct{x} }
	- \sup_{\vct{v} \in K^\polar} \ \ip{ \smash{\vct{u}} }{ \vct{v} } \bigg].
$$
By definition of polarity, the inner supremum takes the value $+\infty$ unless $\vct{u} \in (K^{\polar})^\polar$.
We determine that
$$
\dist( \vct{x}, K^\polar ) 
	\geq \sup_{\vct{u} \in (K^{\polar})^\polar \cap \sphere{d-1}} \ \ip{ \smash{\vct{u}} }{ \vct{x} }
	\geq \sup_{\vct{u} \in K \cap \sphere{d-1}} \ \ip{ \smash{\vct{u}} }{ \vct{x} }.
$$
The last inequality holds because $K \subset (K^{\polar})^\polar$.
\end{proof}

\begin{remark}[Strong duality for cones]
If $K$ is a convex cone and we replace the sphere with a ball, then we have strong duality instead:
$$
\sup_{\vct{u} \in K \cap \ball{d}} \ \ip{ \smash{\vct{u}} }{ \vct{x} }
	= \dist( \vct{x}, K^\polar ).
$$
The proof uses Sion's minimax theorem~\cite{Sio58:General-Minimax}
and the bipolar theorem~\cite[Thm.~14.1]{Roc70:Convex-Analysis}.
\end{remark}

\subsection{The conic Gaussian width of a descent cone}

We can use Proposition~\ref{prop:weak-duality}
to obtain an effective bound for the width of a descent cone.
This approach is based on a classical polarity
correspondence~\cite[Thm.~23.7]{Roc70:Convex-Analysis}.

\begin{fact}[Polarity for descent cones] \label{fact:polar-descent}
The \term{subdifferential} of a proper convex function $f : \R^d \to \overline{\R}$
at a point $\vct{x} \in \R^d$ is the closed convex set
$$
\partial f(\vct{x}) := \big\{ \vct{v} \in \R^d : f(\vct{y}) \geq f(\vct{x}) + \ip{\smash{\vct{y}} - \vct{x}}{ \vct{v} }
	\ \text{for all $\vct{y} \in \R^d$} \big\}.
$$
Assume that the subdifferential $\partial f(\vct{x})$ is nonempty and does not contain the origin.  Then
\begin{equation} \label{eqn:desc-polar}
\Desc(f, \vct{x})^\polar = \overline{\cone}( \partial f(\vct{x}) )
	:= \operatorname{closure}\left( \bigcup_{\tau \geq 0} \tau \cdot \partial f(\vct{x}) \right).
\end{equation}
\end{fact}

\noindent
Combining Proposition~\ref{prop:weak-duality} and Fact~\ref{fact:polar-descent},
we reach a bound for the conic Gaussian width of a descent cone.

\begin{proposition}[The width of a descent cone] \label{prop:width-descent}
Let $f : \R^d \to \overline{\R}$ be a proper convex function, and fix a point $\vct{x} \in \R^d$.
Assume that the subdifferential $\partial f(\vct{x})$ is nonempty and does not contain the origin.  Then
\begin{equation*} w^2\big( \Desc(f, \vct{x}) \big) \leq  
	\Expect{} \inf_{\tau \geq 0} \ \dist^2\big( \vct{g}, \ \tau \cdot \partial f(\vct{x}) \big)
\end{equation*}
\end{proposition}

\noindent
Several specific instances of Proposition~\ref{prop:width-descent} appear in~\cite[App.~C]{CRPW12:Convex-Geometry},
while the general statement here is adapted from~\cite[Sec.~4.1]{ALMT14:Living-Edge-II}.
Sections~\ref{sec:l1-descent} and~\ref{sec:S1-descent}
exhibit how Proposition~\ref{prop:width-descent} works.

\begin{proof}
Proposition~\ref{prop:weak-duality} implies that
$$
w\big( \Desc(f, \vct{x}) \big)
	= \Expect{} \sup_{\vct{u} \in \Desc(f, \vct{x}) \cap \sphere{d-1}} \
	\ip{ \vct{u} }{ \smash{\vct{g}} }
	\leq \Expect{} \dist\big( \vct{g}, \ \Desc(f, \vct{x})^\polar \big).
$$
The expression~\eqref{eqn:desc-polar} for the polar of a descent cone implies that
$$
w\big( \Desc(f, \vct{x}) \big)
	\leq \Expect{} \dist\left( \vct{g}, \
	\operatorname{closure}\left( \bigcup_{\tau \geq 0} \tau \cdot \partial f(\vct{x}) \right) \right)
	= \Expect{} \inf_{\tau \geq 0} \ \dist\big( \vct{g}, \ \tau \cdot \partial f(\vct{x}) \big).
$$
Indeed, the distance to a set is the same as the distance to its closure, and the distance to a union is the
infimal distance to one of its members.  Square the latter display, and apply Jensen's inequality
to complete the argument.
\end{proof}

\subsection{Example: Sparse vectors}
\label{sec:l1-descent}

Suppose that $\vct{x}^\natural$ is a vector in $\R^d$ with $s$ nonzero entries.
Let $\mtx{\Phi}$ be an $m \times d$ matrix whose rows are independent
random vectors distributed as $\normal(\vct{0}, \Id_d)$,
and suppose that we acquire a vector $\vct{y} = \mtx{\Phi} \vct{x}^\natural + \vct{e}$
consisting of $m$ noisy measurements.
We can solve the $\ell_1$-minimization problem~\eqref{eqn:l1-min} in
an attempt to reconstruct $\vct{x}^\natural$.

How many measurements are sufficient to ensure that this approach succeeds?
We will demonstrate that
\begin{equation} \label{eqn:sparse-width}
w^2\big( \Desc(\pnorm{\ell_1}{\cdot}, \vct{x}^\natural) \big)
	\leq 2s \log(d/s) + 2s. \end{equation}
Therefore, Corollary~\ref{cor:recovery-gauss} implies that
$m \gtrsim 2s \log(d/s)$ measurements are enough for us
to recover $\vct{x}^\natural$ approximately.  When $s \ll d$,
the first term in~\eqref{eqn:sparse-width} is numerically sharp
because of~\cite[Prop.~1]{FoyMac:13}.

\subsubsection{The width calculation}

Let us establish the width bound~\eqref{eqn:sparse-width}.
This analysis is adapted from~\cite[App.~C]{CRPW12:Convex-Geometry}
and~\cite[App.~D.2]{ALMT14:Living-Edge-II}; see also~\cite[App.~B]{FoyMac:13}.
The result~\cite[Prop.~4.5]{ALMT14:Living-Edge-II} contains a more complicated
formula for the width that is sharp for all choices of the sparsity $s$.

When estimating widths, a useful strategy is to change coordinates so that the calculations
are more transparent.  The $\ell_1$ norm is invariant under signed permutation, so
$$
\Desc( \pnorm{\ell_1}{\cdot}, \vct{x}^\natural )
	= \mtx{P} \, \Desc( \pnorm{\ell_1}{\cdot}, \mtx{P} \vct{x}^\natural )
	\quad\text{where $\mtx{P}$ is a signed permutation.}
$$
The distribution of a standard Gaussian random variable is invariant under signed permutation,
so the conic Gaussian width has the same invariance.  Therefore,
$$
w\big( \Desc( \pnorm{\ell_1}{\cdot}, \vct{x}^\natural ) \big)
	= w\big( \mtx{P} \, \Desc( \pnorm{\ell_1}{\cdot}, \mtx{P} \vct{x}^\natural ) \big)
	= w\big( \Desc( \pnorm{\ell_1}{\cdot}, \mtx{P} \vct{x}^\natural ) \big).
$$ 
We will use this type of transformation several times without detailed justification.

As a consequence of the argument in the last paragraph, we may assume that $\vct{x}^\natural$ takes the form
$$
\vct{x}^\natural = (x_1, \dots, x_s, 0, \dots, 0)^{\transp} \in \R^d
\quad\text{where}\quad
x_1 \geq \dots \geq x_s > 0.
$$
Proposition~\ref{prop:width-descent} ensures that
\begin{equation} \label{eqn:l1-width-1}
w^2\big( \Desc( \pnorm{\ell_1}{\cdot}, \vct{x}^\natural) \big)
	\leq \Expect{} \dist^2\big( \vct{g}, \ \tau \cdot \partial \pnorm{\ell_1}{\smash{\vct{x}^\natural}} \big)
	\quad\text{for each $\tau \geq 0$}
\end{equation}
where $\vct{g} \sim \normal(\vct{0}, \Id_d)$.
The subdifferential of the $\ell_1$ norm at $\vct{x}^\natural$ satisfies
$$
\partial \pnorm{\ell_1}{\smash{\vct{x}^\natural}}
	= \left\{ \begin{bmatrix} \vct{1}_s \\ \vct{y} \end{bmatrix} \in \R^d : \pnorm{\ell_\infty}{\smash{\vct{y}}} \leq 1 \right\}
	\quad\text{where}\quad
	\vct{1}_s := (1, \dots, 1)^\transp \in \R^s.
$$
Therefore,
\begin{equation} \label{eqn:l1-subdiff-dist}
\Expect{} \dist^2( \vct{g}, \ \tau \cdot \partial \pnorm{\ell_1}{\smash{\vct{x}^\natural}} )
	= \sum_{j = 1}^s \Expect{} \big(g_j - \tau \big)^2
	+ \sum_{j = s+1}^d \Expect{} \big[ \abs{\smash{g_j}} - \tau \big]_+^2.
\end{equation}
As usual, $[a]_+ := \max\{a, 0\}$.  For $1 \leq j \leq s$, a direct calculation gives
\begin{equation} \label{eqn:l1-nonzero-terms}
\Expect \big( g_j - \tau \big)^2
	= 1 + \tau^2.
\end{equation}
For $s < j \leq d$, we apply a familiar tail bound for the standard normal variable to obtain
\begin{equation} \label{eqn:l1-zero-terms}
\Expect \big[ \abs{\smash{g_j}} - \tau \big]_+^2
	= \int_{\tau}^\infty (a - \tau)^2 \, \Prob\big\{ \abs{\smash{g_j}} \geq a \big\} \idiff{a}
	\leq \int_{\tau}^\infty a^2 \left( \sqrt{\frac{2}{\pi}} \, a^{-1} \, \econst^{-a^2/2} \right) \idiff{a}
	< \econst^{-\tau^2/2}.
\end{equation}
Combine~\eqref{eqn:l1-width-1},~\eqref{eqn:l1-subdiff-dist},~\eqref{eqn:l1-nonzero-terms}, and~\eqref{eqn:l1-zero-terms}
to obtain
$$
w^2\big( \Desc( \pnorm{\ell_1}{\cdot}, \vct{x}^\natural) \big)
	\leq \Expect{} \dist^2( \vct{g}, \ \tau \cdot \partial \pnorm{1}{\smash{\vct{x}^\natural}} )
	= s \cdot \big( 1+ \tau^2 \big) + (d-s) \cdot \econst^{-\tau^2/2}.
$$
Choose $\tau^2 = 2\log(d/s)$ and simplify to reach~\eqref{eqn:sparse-width}.

\subsection{Example: Low-rank matrices}
\label{sec:S1-descent}

Let $\mtx{X}^\natural$ be a matrix in $\R^{d_1 \times d_2}$ with rank $r$.
Let $\mtx{\Phi} : \R^{d_1 \times d_2} \to \R^m$ be a linear operator whose
matrix has independent standard Gaussian entries.
Suppose we acquire $m$ noisy measurements of the form
$\vct{y} = \mtx{\Phi}(\mtx{X}^\natural) + \vct{e}$.
We can solve the $S_1$-minimization problem~\eqref{eqn:S1-min}
to reconstruct $\mtx{X}^\natural$.

How many measurements are enough to guarantee that this approach works?
We will prove that
\begin{equation} \label{eqn:rank-width}
w^2\big( \Desc(\pnorm{S_1}{\cdot}, \mtx{X}^\natural) \big)
	\leq 3r\cdot (d_1 + d_2 -r).
\end{equation}
As a consequence, Corollary~\ref{cor:recovery-gauss} implies that
$m \gtrsim 3r \cdot (d_1+ d_2 - r)$ measurements allow us
to identify $\mtx{X}^\natural$ approximately.

\subsubsection{The width calculation}

Let us establish the width bound~\eqref{eqn:rank-width}.
This analysis is adapted from~\cite[App.~C]{CRPW12:Convex-Geometry}
and~\cite[App.~D.3]{ALMT14:Living-Edge-II}; see also~\cite[App.~E]{FoyMac:13}.
The result~\cite[Prop.~4.6]{ALMT14:Living-Edge-II} contains a more complicated
formula for the width that is sharp whenever the rank $r$ is proportional to
the dimension $\min\{d_1, d_2\}$.

The Schatten 1-norm is unitarily invariant,
so we may also select a coordinate system where
$$
\mtx{X}^\natural = \begin{bmatrix} \mtx{\Sigma} & \mtx{0} \\ \mtx{0} & \mtx{0} \end{bmatrix}
	\quad\text{where}\quad
	\mtx{\Sigma} = \operatorname{diag}(\sigma_1, \dots, \sigma_r)
	\quad\text{and}\quad
	\text{$\sigma_j > 0$ for $j = 1, \dots, r$.}
$$
Let $\mtx{G}$ be a $d_1 \times d_2$ matrix with independent standard normal entries, partitioned as
$$
\mtx{G} = \begin{bmatrix} \mtx{G}_{11} & \mtx{G}_{12} \\ \mtx{G}_{21} & \mtx{G}_{22} \end{bmatrix}
	\quad\text{where}\quad
	\text{$\mtx{G}_{11}$ is $r \times r$\quad and\quad $\mtx{G}_{22}$ is $(d_1 - r) \times (d_2 - r)$.}
$$
Define a random parameter $\tau = \norm{ \mtx{G}_{22} }$, where $\norm{\cdot}$ denotes the spectral norm.
Proposition~\ref{prop:width-descent} ensures that
\begin{equation} \label{eqn:S1-width-1}
w^2\big( \Desc(\pnorm{S_1}{\cdot}, \mtx{X}^\natural) \big)
	\leq \Expect{} \dist^2_{\rm F}\big( \vct{G}, \ \tau \cdot \partial \pnorm{S_1}{\smash{\mtx{X}^\natural}} \big).
\end{equation}
Note that we must calculate distance with respect to the Frobenius norm $\fnorm{\cdot}$.
According to~\cite[Ex.~2]{Wat:92},
the subdifferential of the Schatten 1-norm takes the form
$$
\partial \pnorm{S_1}{\smash{\mtx{X}^\natural}} = \left\{ \begin{bmatrix} \Id_r & \mtx{0} \\ \mtx{0} & \mtx{Y} \end{bmatrix} \in \R^{d_1 \times d_2} :
	\norm{ \mtx{Y} }\leq 1 \right\}
	\quad\text{where\quad $\Id_r$ is the $r \times r$ identity matrix.}
$$
We may calculate that
\begin{align} \label{eqn:S1-subdiff-dist}
\Expect{} \dist_{\rm F}^2\big( \mtx{G}, \ \tau \cdot \pnorm{S_1}{\smash{\mtx{X}^\natural}} \big)
	= \Expect{} \fnormsq{ \mtx{G}_{11} - \tau \cdot \Id_r}
	+ \Expect{} \fnormsq{\mtx{G}_{12}} + \Expect{} \fnormsq{\mtx{G}_{21}}
	+ \Expect{} \inf_{\norm{\mtx{Y}} \leq 1}\ \fnormsq{ \mtx{G}_{22} - \tau \cdot \mtx{Y} }.
\end{align}
Our selection of $\tau$ ensures that the last term on the right-hand side of~\eqref{eqn:S1-subdiff-dist}
vanishes.  By direct calculation, \begin{equation} \label{eqn:S1-term-2}
\Expect{} \fnormsq{\mtx{G}_{12}} + \Expect{} \fnormsq{\mtx{G}_{21}}
		= r \cdot (d_1 + d_2 - 2r).
\end{equation}
To bound the first term on right-hand side of~\eqref{eqn:S1-subdiff-dist},
observe that
\begin{equation} \label{eqn:S1-term-1}
\Expect{} \fnormsq{ \mtx{G}_{11} - \tau \cdot \Id_r}
	= r^2 + r \cdot \Expect{} \tau^2
\end{equation}
because the random variable $\tau$ is independent from $\mtx{G}_{11}$.
We need to compute $\Expect{} \tau^2 = \Expect{} \fnormsq{\mtx{G}_{22}}$.  A short argument~\cite[Sec.~2.3]{DS01:Local-Operator}
based on the Slepian comparison inequality shows that
$$
\Expect{} \norm{\mtx{G}_{22}} \leq \sqrt{d_1 - r} + \sqrt{d_2 - r}
	\leq \sqrt{2(d_1 + d_2 - 2r)}.
$$
The spectral norm is 1-Lipschitz, so
the Gaussian Poincar{\'e} inequality~\cite[Thm.~3.20]{BLM13:Concentration-Inequalities} implies
$$
\Expect{} \normsq{\mtx{G}_{22}} - \big( \Expect \norm{\mtx{G}_{22}} \big)^2
 = \Var\big(\norm{\mtx{G}_{22}}\big) \leq 1.
$$
Combining the last two displays,
\begin{equation} \label{eqn:S1-expect-tau}
\Expect{} \tau^2
	= \Expect{} \normsq{\mtx{G}_{22}}
	\leq \big( \Expect{} \norm{\mtx{G}_{22}} \big)^2 + 1
	\leq 2 \, (d_1 + d_2 - 2r) + 1.
\end{equation}
Finally, we incorporate~\eqref{eqn:S1-subdiff-dist},~\eqref{eqn:S1-term-1},~\eqref{eqn:S1-term-2}, and~\eqref{eqn:S1-expect-tau} into the width bound~\eqref{eqn:S1-width-1} to reach
$$
w^2\big( \Desc( \pnorm{S_1}{\cdot}, \mtx{X}^\natural ) \big)
	\leq 3r \cdot (d_1 + d_2 - 2r) + r^2 + r.
$$
Simplify this expression to obtain the result~\eqref{eqn:rank-width}.

\section{Mendelson's Small Ball Method}
\label{sec:beyond-gauss}

In Sections~\ref{sec:convex-recovery}--\ref{sec:width-descent},
we analyzed a convex programming method
for recovering structured signals from standard Gaussian measurements.
The main result, Corollary~\ref{cor:recovery-gauss}, is appealing because it applies to any convex
complexity measure $f$.  Proposition~\ref{prop:width-descent} allows
us to instantiate this result because it provides a mechanism for controlling
the Gaussian width of a descent cone.  On the other hand, this approach only works
when the sampling matrix $\mtx{\Phi}$ follows the
standard Gaussian distribution.

For other sampling models, researchers use a variety of ad hoc techniques to study
the recovery problem.  It is common to see a separate and intricate argument for
each new complexity measure $f$ and each new distribution for $\mtx{\Phi}$.
It is natural to wonder whether there is a single approach that can address
a broad class of complexity measures and sampling matrices.

The primary goal of this chapter is to analyze convex signal reconstruction
with more general random measurements.  Our argument is based on Mendelson's
\term{Small Ball Method}, a powerful strategy for establishing a lower
bound on a nonnegative empirical process~\cite{KM13:Bounding-Smallest,Men13:Remark-Diameter,Men14:Learning-Concentration,LM14:Compressed-Sensing,Men14:Learning-Concentration-II}.
This section contains an overview of Mendelson's Small Ball Method.
Section~\ref{sec:subgauss-model} uses this technique to study subgaussian measurement models.
In Section~\ref{sec:bowling}, we extend these ideas to a larger class of sampling distributions.
In Section~\ref{sec:phase}, we conclude with an application to the problem of
phase retrieval.

\subsection{The minimum conic singular value as a nonnegative empirical process}
\label{sec:conic-sing-empirical}

Suppose that $\vct{\phi}$ is a random vector on $\R^d$, and draw independent copies $\vct{\phi}_1, \dots, \vct{\phi}_m$ of the random vector $\vct{\phi}$.  Form an $m \times d$ sampling matrix $\mtx{\Phi}$ whose rows are these random vectors:
\begin{equation} \label{eqn:indep-rows}
\mtx{\Phi} = \begin{bmatrix} & \vct{\phi}_1^\transp & \\ & \vdots & \\ & \vct{\phi}_m^\transp & \end{bmatrix}.
\end{equation}
Fix a cone $K \in \R^d$, not necessarily convex, and define the set $E := K \cap \sphere{d-1}$.  Then we can express the minimum conic singular value
$\lambda_{\min}(\mtx{\Phi}; K)$ of the sampling matrix as a nonnegative empirical process:
\begin{equation} \label{eqn:conic-sing-sum}
\lambda_{\min}(\mtx{\Phi}; K)
	= \inf_{\vct{u} \in E} \ \left( \sum_{i=1}^m \abssqip{\vct{u}}{\smash{\vct{\phi}_i}} \right)^{1/2}.
\end{equation}
When the sampling matrix is Gaussian, we can use Gordon's theorem~\cite[Thm.~1.4]{Gor85:Some-Inequalities}
to obtain a lower bound for the expression~\eqref{eqn:conic-sing-sum}, as in Proposition~\ref{prop:conic-sing-gauss}.
The challenge is to find an alternative method for producing
a lower bound in a more general setting.

\subsection{A lower bound for nonnegative empirical processes}

The main technical component in Mendelson's Small Ball Method is a remarkable estimate that
was developed in the paper~\cite{Men14:Learning-Concentration}.  This result delivers an
effective lower bound for a nonnegative empirical process.

\begin{proposition}[Lower bound for a nonnegative empirical process~\protect{\cite[Thm.~5.4]{Men14:Learning-Concentration}}] \label{prop:km}
Fix a set $E \subset \R^d$.  Let $\vct{\phi}$ be a random vector on $\R^d$, and let $\vct{\phi}_1, \dots, \vct{\phi}_m$ be independent copies of $\vct{\phi}$.  Define the $m \times d$ matrix $\mtx{\Phi}$ as in~\eqref{eqn:indep-rows}.  Introduce the marginal tail function
$$
Q_\xi(E; \vct{\phi}) := \inf_{\vct{u} \in E} \
\Prob\big\{ \absip{\vct{u}}{\smash{\vct{\phi}}} \geq \xi \big\}
\quad\text{where $\xi \geq 0$.}
$$
Let $\eps_1, \dots, \eps_m$ be independent Rademacher random variables,\footnote{A \term{Rademacher} random variable takes the two values $\pm 1$ with equal probability.}
independent from everything else, and define the mean empirical width of the set:
\begin{equation} \label{eqn:mean-emp-width}
W_{m}(E; \vct{\phi}) := \Expect{} \sup_{\vct{u} \in E} \ \ip{\vct{u}}{\vct{h}}
\quad\text{where}\quad
\vct{h} := \frac{1}{\sqrt{m}} \sum_{i=1}^m \eps_i \vct{\phi}_i.
\end{equation}
Then, for any $\xi > 0$ and $t > 0$,
$$
	\inf_{\vct{u} \in E } \ \left( \sum_{i=1}^m \abssqip{\vct{u}}{ \smash{\vct{\phi}_i }}
	\right)^{1/2} \geq \xi \sqrt{m} \,
	Q_{2\xi}(E; \vct{\phi}) - 2 W_{m}(E; \vct{\phi}) - \xi t
$$
with probability at least $1 - \econst^{-t^2 / 2}$.
\end{proposition}

\noindent
The proof appears below in Section~\ref{sec:km}.  In the sequel, we usually lighten our notation for $Q_\xi$ and $W_m$ by suppressing the dependence on $\vct{\phi}$.

Before we continue, it may be helpful to remark on this result.
The marginal tail function $Q_\xi(E)$ reflects the probability that the random variable $\absip{\vct{u}}{\smash{\vct{\phi}}}$ is close to zero for any fixed vector $\vct{u} \in E$.  When $Q_\xi(E)$ is bounded away from zero for some $\xi$, the nonnegative empirical process is likely to be large.  Koltchinskii \& Mendelson~\cite{KM13:Bounding-Smallest} point out that the marginal tail function reflects the absolute continuity of the distribution of $\vct{\phi}$, so $Q_\xi$ may be quite small when $\vct{\phi}$ is ``spiky.''

The mean empirical width $W_m(E)$ is a distribution-dependent measure of the size of the set $E$.  When $\vct{\phi}$ follows a standard Gaussian distribution, $W_m(E)$ reduces to the usual Gaussian width $W(E) := \Expect{} \sup_{\vct{u} \in E} \ \ip{\vct{u}}{\smash{\vct{g}}}$.  As the number $m$ tends to infinity, the distribution of the random vector $\vct{h}$ converges in distribution to a centered Gaussian variable with covariance $\Expect[ \vct{\phi} \vct{\phi}^\adj ]$.  Therefore, $W_m(E) \to W(E)$ when $\vct{\phi}$ is centered and isotropic.

\subsection{Mendelson's Small Ball Method}

Proposition~\ref{prop:km} shows that we can obtain a lower bound for~\eqref{eqn:conic-sing-sum}
by performing two simpler estimates.  To achieve this goal, Mendelson has developed a general strategy,
which consists of three steps:

\begin{center}
\begin{table}[h!]
\framebox{
\begin{minipage}{0.9\textwidth}
\vspace{0.5pc}
\centering{\textsc{\textbf{Mendelson's Small Ball Method}}} \vspace{0.5pc}
\begin{enumerate} \setlength{\itemsep}{2mm}
\item	Apply Proposition~\ref{prop:km} to bound the
minimum conic singular value $\lambda_{\min}\big( \mtx{\Phi}; \ K \big)$ below in terms of
the marginal tail function $Q_{2\xi}(E; \vct{\phi} )$
and the mean empirical width $W_m( E; \vct{\phi} )$.
The index set $E := K \cap \sphere{d-1}$.

\item	Bound the marginal tail function $Q_{2\xi}(E; \vct{\phi})$ below using a Paley--Zygmund inequality.

\item	Bound the mean empirical width $W_m(E; \vct{\phi})$ above by imitating techniques for
controlling the Gaussian width of $E$.
\end{enumerate}
\vspace{1pc}
\end{minipage}
\hspace{0.25in}}
\end{table}
\end{center}

\noindent
This presentation is distilled from the corpus~\cite{KM13:Bounding-Smallest,Men13:Remark-Diameter,Men14:Learning-Concentration,LM14:Compressed-Sensing,Men14:Learning-Concentration-II}.  A more sophisticated variant of this
method appears in~\cite[Thm.~5.3]{Men14:Learning-Concentration}.  Later in this chapter,
we will encounter several concrete applications of this strategy.

\subsection{Expected Scope}

Mendelson's Small Ball Method provides lower bounds for~\eqref{eqn:conic-sing-sum} in many
situations, but it does not offer a universal prescription.  Let us try to delineate
the circumstances where this approach is likely to be useful for signal recovery problems.

\begin{itemize} \setlength{\itemsep}{2mm}
\item	Mendelson's Small Ball Method assumes that the sampling matrix $\mtx{\Phi}$ has independent, identically
distributed rows.  Although this model describes many of the sampling strategies in the literature,
there are some examples, such as random filtering~\cite{TWDBB06:Random-Filters}, that do not conform
to this assumption.

\item	A major advantage of Mendelson's Small Ball Method is that it applies to sampling distributions with heavy
tails.  On the other hand, the random vector $\vct{\phi}$ cannot be too ``spiky,'' or else it may
not be possible to produce a good lower bound for the marginal tail function $Q_{2\xi}(E)$.
This requirement indicates that the approach may require significant
improvements before it applies to problems like matrix completion.
\end{itemize}

\noindent
There are a number of possible extensions of Mendelson's Small Ball Method that could expand its bailiwick.
For example, it is easy to extend Proposition~\ref{prop:km}
to address the case where the random vector $\vct{\phi}$ is complex-valued.
A more difficult, but very useful, modification would allow us to block the measurements into groups.
This revision could reduce the difficulties associated with spiky distributions, but it seems to demand
some additional ideas.

\subsection{Proof of Proposition~\ref{prop:km}}
\label{sec:km}

Let us establish the Mendelson bound for a nonnegative empirical process.  First, we introduce a directional version of the marginal tail function:
$$
Q_\xi(\vct{u}) := \Prob\big\{ \absip{\vct{u}}{\smash{\vct{\phi}}} \geq \xi \big\}
\quad\text{for $\vct{u} \in E$ and $\xi > 0$.}
$$
Lyapunov's inequality and Markov's inequality give the numerical bounds
$$
\left( \frac{1}{m} \sum_{i=1}^m \abssqip{\vct{u}}{\smash{\vct{\phi}_i}} \right)^{1/2}
	\geq \frac{1}{m} \sum_{i=1}^m \absip{\vct{u}}{\smash{\vct{\phi}_i}}
	\geq \frac{\xi}{m} \sum_{i=1}^m \mathbb{1} \big\{ \absip{\vct{u}}{\smash{\vct{\phi}_i}} \geq \xi \big\}.
$$
We write $\mathbb{1}A$ for the 0--1 random variable that indicates whether the event $A$ takes place.  Add and subtract $Q_{2\xi}(\vct{u})$ inside the sum, and then take the infimum over $\vct{u} \in E$ to reach the inequality
\begin{equation} \label{eqn:inf-sure-bd}
\inf_{\vct{u} \in E} \ \left( \frac{1}{m} \sum_{i=1}^m \abssqip{\vct{u}}{\smash{\vct{\phi}_i}}
	\right)^{1/2}
	\geq \xi \inf_{\vct{u} \in E} Q_{2\xi}(\vct{u}) - \frac{\xi}{m} \sup_{\vct{u} \in E} \
	\sum_{i=1}^m \big[ Q_{2\xi}(\vct{u}) -
	\mathbb{1}\big\{ \absip{\vct{u}}{\smash{\vct{\phi}_i} } \geq \xi \big\} \big].
\end{equation}
To control the supremum in probability, we can invoke the bounded difference inequality~\cite[Sec.~6.1]{BLM13:Concentration-Inequalities}.  Observe that each summand is independent and bounded in magnitude by one.  Therefore,
\begin{equation} \label{eqn:sup-prob-bd}
\sup_{\vct{u} \in E}\ \sum_{i=1}^m \big[ Q_{2\xi}(\vct{u}) -
	\mathbb{1}\big\{ \absip{\vct{u}}{\smash{\vct{\phi}_i} } \geq \xi \big\} \big] \\
	\leq \Expect{} \sup_{\vct{u} \in E} \
	\sum_{i=1}^m \big[ Q_{2\xi}(\vct{u}) -
	\mathbb{1}\big\{ \absip{\vct{x}}{\smash{\vct{\phi}_i} } \geq \xi \big\} \big]
	+ t \sqrt{m}
\end{equation}
with probability at least $1 - \econst^{- t^2 / 2}$.

Next, we simplify the expected supremum.  Introduce a soft indicator function:
$$
\psi_\xi : \R \to [0,1]
\quad\text{where}\quad
\psi_\xi(s) := \begin{cases}
	0, & \abs{s} \leq \xi \\
	(\abs{s}-\xi)/\xi, & \xi < \abs{s} \leq 2\xi \\
	1, & 2\xi < \abs{s}.
\end{cases}
$$
We need two properties of the soft indicator.
First, the soft indicator is bracketed by two hard indicators: $\mathbb{1}\{\abs{s} \geq 2\xi \} \leq \psi_\xi(s) \leq \mathbb{1}\{\abs{s} \geq \xi \}$ for all $s \in \R$.  Second, $\xi \psi_\xi$ is a \term{contraction}, i.e., a 1-Lipschitz function on $\R$ that fixes the origin.
Therefore, we can make the following calculation:
\begin{align} \label{eqn:sup-expect-bd}
\Expect{} \sup_{\vct{u} \in E} \
	\sum_{i=1}^m \big[ Q_{2\xi}(\vct{u}) -
	\mathbb{1}\big\{ \absip{\vct{u}}{\smash{\vct{\phi}_i} } \geq \xi \big\} \big]
&= \Expect{} \sup_{\vct{u} \in E} \
	\sum_{i=1}^m \left[ \Expect \mathbb{1}\big\{ \absip{\vct{u}}{\smash{\vct{\phi}} } \geq 2\xi \big\} -
	\mathbb{1}\big\{ \absip{\vct{u}}{\smash{\vct{\phi}_i} } \geq \xi \big\} \right] \notag \\
&\leq \Expect{} \sup_{\vct{u} \in E} \
	\sum_{i=1}^m \left[ \Expect \psi_\xi( \ip{\vct{u}}{\smash{\vct{\phi}}} ) 
	- \psi_\xi( \ip{\vct{u}}{\smash{\vct{\phi}_i}} ) \right] \notag \\
&\leq 2 \Expect{} \sup_{\vct{u} \in E} \
	\sum_{i=1}^m \eps_i \psi_\xi(\ip{\vct{u}}{\smash{\vct{\phi}_i}}) \notag \\
&\leq \frac{2}{\xi} \Expect{} \sup_{\vct{u} \in E} \
	\sum_{i=1}^m \eps_i \ip{\vct{u}}{\smash{\vct{\phi}_i}}. 
\end{align}
In the first line, we write the marginal tail function as an expectation, and then we bound the two indicators using the soft indicator function.  The next inequality is the Gin{\'e}--Zinn symmetrization~\cite[Lem.~2.3.1]{VW96:Weak-Convergence}.  The last line follows from the Rademacher comparison principle~\cite[Eqn.~(4.20)]{LT91:Probability-Banach} because $\xi \psi_\xi$ is a contraction.

Combine the inequalities~\eqref{eqn:inf-sure-bd},~\eqref{eqn:sup-prob-bd}, and~\eqref{eqn:sup-expect-bd} to reach
$$
\inf_{\vct{u} \in E} \ \left( \frac{1}{m} \sum_{i=1}^m \abssqip{\vct{u}}{\smash{\vct{\phi}_i}}
	\right)^{1/2}
	\geq \xi \inf_{\vct{u} \in E} Q_{2\xi}(\vct{u}) - \frac{\xi}{m}
	\left[ \frac{2}{\xi} \Expect{} \sup_{\vct{u} \in E} \
	\sum_{i=1}^m \eps_i \ip{\vct{u}}{\smash{\vct{\phi}_i}} + t \sqrt{m} \right].
$$
Define $\vct{h} := m^{-1/2} \sum_{i=1}^m \eps_i \vct{\phi}_i$, and clear the factor $\sqrt{m}$ to conclude that
$$
\inf_{\vct{u} \in E} \ \left( \sum_{i=1}^m \abssqip{\vct{u}}{\smash{\vct{\phi}_i}}
	\right)^{1/2}
	\geq \xi \sqrt{m} \inf_{\vct{u} \in E} Q_{2\xi}(\vct{u})
	- 2 \Expect{} \sup_{\vct{u} \in E} \ \ip{\vct{u}}{\vct{h}} - \xi t.
$$
with probability at least $1 - \econst^{-t^2/2}$.  Identify the marginal tail function $Q_{2\xi}(E)$ and the empirical width $W_m(E)$ to establish Proposition~\ref{prop:km}.

\section{A universal error bound for subgaussian measurements}
\label{sec:subgauss-model}

In this section, we invoke Mendelson's Small Ball Method to study convex signal
recovery from independent subgaussian measurements.  This class of examples
provides a wide generalization of standard Gaussian measurements.  We will establish
a variant of the Gaussian recovery result, Corollary~\ref{cor:recovery-gauss},
in this setting.

\subsection{Subgaussian measurements}
\label{sec:subgauss}

Let us set out the conditions we require for the sampling matrix.
Suppose that $\vct{\phi}$ is a random vector in $\R^d$ that has the following properties.

\begin{itemize} \setlength{\itemsep}{2mm}
\item	\textbf{[Centering]}  The vector has zero mean: $\Expect{} \vct{\phi} = \vct{0}$.

\item	\textbf{[Nondegeneracy]}  There is a positive constant $\alpha$ for which 
$$
\alpha \leq \Expect{} \absip{ \vct{u} }{ \smash{\vct{\phi}} }
\quad\text{for each $\vct{u} \in \sphere{d-1}$.}
$$

\item	\textbf{[Subgaussian marginals]}  There is a positive constant $\sigma$ for which
$$
\Prob\big\{ \absip{ \vct{u} }{ \smash{\vct{\phi}} } \geq t \big\}
	\leq 2 \econst^{-t^2/(2\sigma^2)}
	\quad\text{for each $\vct{u} \in \sphere{d-1}$.}
$$

\item	\textbf{[Low eccentricity]}  The eccentricity $\rho := \sigma / \alpha$ of the distribution should be small.
\end{itemize}

\noindent
Finally, we construct a random $m \times d$ sampling matrix $\mtx{\Phi}$
whose rows are independent copies of $\vct{\phi}^\transp$,
as in the expression~\eqref{eqn:indep-rows}.

A few examples of subgaussian distributions may be helpful.

\begin{example}[Nonstandard Gaussian matrices]
Suppose that $\vct{\phi} \in \R^d$ follows the $\normal(\vct{0}, \mtx{\Sigma})$ distribution where the covariance $\mtx{\Sigma}$ satisfies $\tfrac{\pi}{2} \alpha^2 \leq \vct{u}^\transp \mtx{\Sigma} \vct{u} \leq \sigma^2$  for each vector $\vct{u} \in \sphere{d-1}$.  Then the required conditions follow from basic facts about a normal distribution.
\end{example}

\begin{example}[Independent bounded entries]
Let $X$ be a symmetric random variable whose magnitude is bounded by $\sigma$.  Suppose that each entry of $\vct{\phi}$ is an independent copy of $X$.

The vector $\vct{\phi}$ inherits centering from $X$.  Next, $\vct{\phi}$ is nondegenerate with $\alpha \geq 2^{-1/2} \Expect \abs{X}$ because of the Khintchine inequality~\cite{LO94:Best-Constant} and a convexity argument.  Finally, $\vct{\phi}$ has subgaussian marginals with the parameter $\sigma$ because of Hoeffding's inequality~\cite[Sec.~2.6]{BLM13:Concentration-Inequalities}.
\end{example}

\subsection{The minimum conic singular value of a subgaussian matrix}
\label{sec:conic-sing-subgauss}

The main result of this section gives a lower bound for the minimum conic singular value of a matrix $\mtx{\Phi}$ that satisfies the conditions in Section~\ref{sec:subgauss}.

\begin{theorem}[Minimum conic singular value of a subgaussian matrix] \label{thm:conic-sing-subgauss}
Suppose $\mtx{\Phi}$ is an $m \times d$ random matrix that satisfies the conditions in Section~\ref{sec:subgauss}.
Let $K \subset \R^d$ be a cone, not necessarily convex.  Then
$$
\lambda_{\min}(\mtx{\Phi}; K) \geq c \alpha \rho^{-2} \cdot \sqrt{m} - C \sigma \cdot w(K) - \alpha t
$$
with probability at least $1 - \econst^{-ct^2}$.
The quantities $c$ and $C$ are positive absolute constants.
\end{theorem}

\noindent
Observe that, when the eccentricity $\rho$ has constant order, the bound in Theorem~\ref{thm:conic-sing-subgauss} matches the result for Gaussian matrices in Proposition~\ref{prop:conic-sing-gauss}.  A similar result appears in the paper~\cite{MPT07:Reconstruction-Subgaussian}, so we do not claim any novelty.
We establish Theorem~\ref{thm:conic-sing-subgauss} below in Section~\ref{sec:conic-sing-subgauss}.

\subsection{An error bound for subgaussian measurements}

Combining Proposition~\ref{prop:deterministic-error} and Theorem~\ref{thm:conic-sing-subgauss}, we reach an immediate consequence for signal recovery from subgaussian measurements.

\begin{corollary}[Signal recovery from subgaussian measurements]
\label{cor:recovery-subgauss}
Let $\vct{x}^\natural$ be a signal in $\R^d$.  Let $\mtx{\Phi}$ be an $m \times d$ random matrix that satisfies the conditions in Section~\ref{sec:subgauss}, and let $\vct{y} = \mtx{\Phi} \vct{x}^\natural + \vct{e}$ be a vector of measurements in $\R^m$.  With probability at least $1 - \econst^{-ct^2}$, the following statement holds.  Assume that $\enorm{\vct{e}} \leq \eta$, and let $\widehat{\vct{x}}_{\eta}$ be any solution to the optimization problem~\eqref{eqn:convex-recovery}.  Then
$$
\smnorm{}{ \widehat{\vct{x}}_{\eta} - \vct{x}^\natural } \leq \frac{2\eta}{\big[ c \alpha \rho^{-2} \cdot \sqrt{m} - C \sigma \cdot w\big(\Desc(f, \vct{x}^\natural) \big) - \alpha t \big]_+ }.
$$
The quantities $c$ and $C$ are positive absolute constants.  The operation $[a]_+ := \max\{a,0\}$ returns the positive part of a number.
\end{corollary}

Corollary~\ref{cor:recovery-subgauss} provides for stable recovery of $\vct{x}^\natural$ as soon as the number $m$ of subgaussian measurements satisfies
$$
m \geq C' \rho^{6} \cdot w^2\big( \Desc(f,\vct{x}^\natural) \big).
$$
How accurate is this result?  Note that standard Gaussian measurements satisfy the assumptions of the corollary with $\rho$ constant, and we need at least $w^2\big(\Desc(f, \vct{x}^\natural) \big)$ standard normal measurements to recover the structured signal $\vct{x}^\natural$ with the complexity measure $f$.  Therefore, the bound is correct up to the constant factor $C'$ and the precise dependence on the eccentricity $\rho$.

\subsection{Proof of Theorem~\ref{thm:conic-sing-subgauss}: Setup and Step 1}
\label{sec:subgauss2}

To establish Theorem~\ref{thm:conic-sing-subgauss}, we rely on Mendelson's Small Ball Method.  The argument also depends on some deep ideas from the theory of generic chaining~\cite{Tal00:Generic-Chaining}, but we only use these results in a na{\"i}ve way.

Fix a cone $K$ in $\R^d$, and define the set $E := K \cap \sphere{d-1}$.  Suppose that $\vct{\phi}$ is a random vector in $\R^d$ that satisfies the conditions set out in Section~\ref{sec:subgauss}, and construct an $m \times d$ random matrix $\mtx{\Phi}$ whose rows are independent copies of $\vct{\phi}$.
Proposition~\ref{prop:km} implies that
\begin{equation} \label{eqn:conic-sing-bd1}
\lambda_{\min}(\mtx{\Phi}; K)
	\geq \xi \sqrt{m} \, Q_{2\xi}(E) - 2 W_m(E) - \xi t
	\quad\text{with probability $\geq 1 - \econst^{-t^2/2}$.}
\end{equation}
This result holds for all $\xi > 0$ and $t > 0$.  To establish Theorem~\ref{thm:conic-sing-subgauss}, we must develop a constant lower bound for the marginal tail function $Q_{2\xi}(E)$, and we also need to compare the mean empirical width $W_m(E)$ with the conic Gaussian width $w(K)$.

\subsection{Step 2: The marginal tail function}

We begin with the lower bound for the marginal tail function $Q_{2\xi}$.  This result is an easy consequence of the second moment method, also known as the Paley--Zygmund inequality.
Let $\vct{u}$ be any vector in $E$.  One version of the second moment method states that
\begin{equation} \label{eqn:paley-zygmund}
\Prob\big\{ \absip{ \vct{u} }{ \smash{\vct{\phi}} } \geq 2\xi \big\}
	\geq \frac{\big[\Expect{} \absip{\vct{u}}{\smash{\vct{\phi}}} - 2\xi  \big]_{+}^2}{\Expect{} \absip{\vct{u}}{\smash{\vct{\phi}}}^2}.
\end{equation}
To control the denominator on the right-hand side of~\eqref{eqn:paley-zygmund},
we use the subgaussian marginal condition to estimate that
$$
\Expect{} \absip{\vct{u}}{\smash{\vct{\phi}}}^2
	= \int_0^\infty 2 s \cdot \Prob\big\{ \absip{\vct{u}}{\smash{\vct{\phi}}} \geq s \big\} \idiff{s}
	\leq 4 \sigma^2.
$$
To bound the numerator on the right-hand side of~\eqref{eqn:paley-zygmund},
we use the nondegeneracy assumption: $\Expect{} \absip{\vct{u}}{\smash{\vct{\phi}}} \geq \alpha$.  
Combining these results and taking the infimum over $\vct{u} \in E$, we reach
\begin{equation} \label{eqn:QE-lower}
Q_{2\xi}(E)
	= \inf_{\vct{u} \in E} \ \Prob\big\{ \absip{ \vct{u} }{ \smash{\vct{\phi}} } \geq 2\xi \big\}
	\geq \frac{(\alpha - 2\xi)^2}{4\sigma^2}
\end{equation}
for any $\xi$ that satisfies $2 \xi < \alpha$.

\subsection{Step 3: The mean empirical width}

Next, we demonstrate that the empirical width $W_m(E)$ is controlled by the conic Gaussian width $w(K)$.  This argument requires sophisticated results from the theory of generic chaining~\cite{Tal00:Generic-Chaining}.  First, observe that the vector $\vct{h} = m^{-1/2} \sum\nolimits_{i=1}^m \eps_i \vct{\phi}_i$ inherits subgaussian marginals from the centered subgaussian distribution $\vct{\phi}$.  Indeed,
$$
\Prob\big\{ \absip{ \vct{u} }{ \vct{h} } \geq t \big\}
	\leq C_1 \econst^{-c_1 t^2/\sigma^2}
	\quad\text{for each $\vct{u} \in \sphere{d-1}$.}
$$
See~\cite[Sec.~5.2.3]{Ver11:Introduction-Nonasymptotic} for an introduction to subgaussian random variables.  In particular, we have the bound
$$
\Prob\big\{ \absip{ \vct{u} - \vct{v} }{ \vct{h} } \geq t \big\}
	\leq C_1 \econst^{-c_1 t^2/(\sigma^2 \enormsm{\vct{u} - \vct{v}}^2)}
	\quad\text{for all $\vct{u}, \vct{v} \in \R^d$.}
$$
Under the latter condition, the generic chaining theorem~\cite[Thm.~1.2.6]{Tal00:Generic-Chaining} asserts that
$$
W_m(E) = \Expect{} \sup_{\vct{u} \in E} \ \ip{ \vct{u} }{ \vct{h} }
	\leq C_2 \sigma \cdot \gamma_2( E, \ell_2 )
$$
where $\gamma_2$ is a geometric functional.
The precise definition of $\gamma_2$ is not important for our purposes because the majorizing measure theorem~\cite[Thm.~2.1.1]{Tal00:Generic-Chaining} states that
$$
\gamma_2(E, \ell_2) \leq C_3 \cdot \Expect{} \sup_{\vct{u} \in E} \ \ip{ \vct{u} }{ \smash{\vct{g}} }
$$
where $\vct{g} \sim \normal(\vct{0}, \Id_d)$.  It follows that
\begin{equation} \label{eqn:emp-width-upper}
W_m(E) 	\leq C_4 \sigma \cdot \Expect{} \sup_{\vct{u} \in E} \ \ip{ \vct{u} }{ \smash{\vct{g}} }
	= C_4 \sigma \cdot w(K).
\end{equation}
We have recalled that $E = K \cap \sphere{d-1}$ to identify the conic Gaussian width $w(K)$.

\subsection{Combining the bounds}

Combine the bounds~\eqref{eqn:conic-sing-bd1},~\eqref{eqn:QE-lower}, and~\eqref{eqn:emp-width-upper} to discover that
$$
\lambda_{\min}(\mtx{\Phi}; K)
	\geq \xi \sqrt{m} \cdot \frac{(\alpha - 2\xi)^2}{4\sigma^2}
	- 2 C_4 \sigma \, w(K) - \xi t
	\quad\text{with probability $\geq 1 - \econst^{-t^2/2}$,}
$$
provided that $2\xi < \alpha$.  Select $\xi = \alpha/6$ to see that
\begin{equation} \label{eqn:tail-bound}
\lambda_{\min}(\mtx{\Phi}; K)
	\geq \frac{1}{54} \cdot \frac{\alpha^3}{\sigma^2} \sqrt{m} - C_5 \sigma \, w(K)
	- \frac{\alpha}{6} t
	\quad\text{with probability $\geq 1 - \econst^{-t^2/2}$.}
\end{equation}
Using the eccentricity $\rho = \sigma/\alpha$,
we simplify the expression~\eqref{eqn:tail-bound} to reach a bound for the minimum conic singular value of a subgaussian random matrix $\mtx{\Phi}$ that satisfies the conditions set out in Section~\ref{sec:subgauss}.  This completes the proof of Theorem~\ref{thm:conic-sing-subgauss}.

\section{The bowling scheme}
\label{sec:bowling}

As we have seen in Theorem~\ref{thm:conic-sing-subgauss},
subgaussian sampling models exhibit behavior similar with the standard Gaussian measurement model.
Yet there are many interesting problems where the random sampling matrix does not conform to the subgaussian
assumption.  In this section, we explain how to adapt Mendelson's Small Ball Method to a range of other
sampling ensembles.  The key idea is to use the conic duality arguments from Section~\ref{sec:width-descent}
to complete the estimate for the mean empirical width.

\subsection{The mean empirical width of a descent cone}
\label{sec:emp-width-descent}

Let us state a simple duality result for the mean empirical width of a descent cone.  This bound
is based on the same principles as Proposition~\ref{prop:width-descent}.

\begin{proposition}[The mean empirical width of a descent cone] \label{prop:emp-width-descent}
Let $f : \R^d \to \overline{\R}$ be a proper convex function, and fix a point $\vct{x} \in \R^d$.
Assume that the subdifferential $\partial f(\vct{x})$ is nonempty and does not contain the origin.
For any random vector $\vct{\phi} \in \R^d$,
\begin{equation*} W_m\big( \Desc(f, \vct{x}) \cap \sphere{d-1}; \vct{\phi} \big) \leq  
	\Expect{} \inf_{\tau \geq 0} \ \dist^2\big( \vct{h}, \ \tau \cdot \partial f(\vct{x}) \big)
	\quad\text{where}\quad
	\vct{h} := \frac{1}{\sqrt{m}} \sum_{i=1}^m \eps_i \vct{\phi}_i.
\end{equation*}
The mean empirical width $W_m$ is defined in~\eqref{eqn:mean-emp-width}.
The random vectors $\vct{\phi}_1, \dots, \vct{\phi}_m$ are independent copies of $\vct{\phi}$,
and $\eps_1, \dots, \eps_m$ are independent Rademacher random variables.
\end{proposition}

\begin{proof}
The argument is identical with the proof of Proposition~\ref{prop:width-descent}
once we replace the Gaussian vector $\vct{g}$ with the random vector $\vct{h}$.
\end{proof}

\subsection{The bowling scheme}

We are now prepared to describe a general approach for convex signal recovery
from independent random measurements.

The setup is similar with previous sections.
Consider an unknown structured signal $\vct{x}^\natural \in \R^d$
and a complexity measure $f : \R^d \to \overline{\R}$ that is proper and convex.
Let $\mtx{\Phi}$ be a known $m \times d$ sampling matrix,
and suppose that we acquire $m$ noisy linear measurements
of the form $\vct{y} = \mtx{\Phi} \vct{x}^\natural + \vct{e}$.
We wish to analyze the performance of the convex recovery method~\eqref{eqn:convex-recovery}.
Proposition~\ref{prop:deterministic-error} shows that we can accomplish this goal by finding
a lower bound for the minimum conic singular value of the descent cone:
\begin{equation} \label{eqn:conic-sing-???}
\lambda_{\min}\big( \mtx{\Phi}; \ \Desc(f, \vct{x}^\natural) \big)
	\geq \underline{\hspace{2pc} ??? \hspace{2pc}}.
\end{equation}

We want to produce a bound of the form~\eqref{eqn:conic-sing-???} when the rows of the
measurement matrix $\mtx{\Phi}$ are independent copies of a random vector $\vct{\phi}$.
This problem falls within the scope of Mendelson's Small Ball Method.
Introduce the index set $E := \Desc(f, \vct{x}^{\natural}) \cap \sphere{d-1}$.
In light of~\eqref{eqn:conic-sing-sum},
$$
\lambda_{\min}\big( \mtx{\Phi}; \ \Desc(f, \vct{x}^\natural) \big)
	= \inf_{\vct{u} \in E} \ \left( \sum_{i=1}^m \abssqip{\vct{u}}{\vct{\phi}_i} \right)^{1/2}.
$$
We follow Mendelson's general strategy to control the minimum conic singular value,
but we propose a specific technique for bounding the mean empirical width that exploits
the structure of the index set $E$. 

\begin{center}
\begin{table}[h!]
\framebox{
\begin{minipage}{0.9\textwidth}
\vspace{0.5pc}
\centering{\textsc{\textbf{the bowling scheme}}} \vspace{0.5pc}
\begin{enumerate} \setlength{\itemsep}{2mm}
\item	Apply Proposition~\ref{prop:km} to bound the
minimum conic singular value $\lambda_{\min}\big( \mtx{\Phi}; \ \Desc(f, \vct{x}^{\natural} \big)$ below in terms of
the marginal tail function $Q_{2\xi}(E; \vct{\phi} )$
and the mean empirical width $W_m( E; \vct{\phi} )$.
The index set $E := \Desc(f; \vct{x}^\natural) \cap \sphere{d-1}$.

\item	Bound the marginal tail function $Q_{2\xi}(E; \vct{\phi})$ below using a Paley--Zygmund inequality.

\item[(3$'$)]  Apply Proposition~\ref{prop:emp-width-descent} to control the
mean empirical width $W_m(E; \vct{\phi})$.
\end{enumerate}
\vspace{1pc}
\end{minipage}
\hspace{0.25in}}
\end{table}
\end{center}

\noindent
In other words, Step (3) of Mendelson's framework has been specialized to Step (3$'$).

We refer to this instance of Mendelson's Small Ball Method as \term{the bowling scheme}.
The name is chosen as a salute to David Gross's \term{golfing scheme}.  Whereas the
golfing scheme is based on dual optimality conditions for the signal recovery problem~\eqref{eqn:convex-recovery},
the bowling scheme is based on the primal optimality condition through Proposition~\ref{prop:deterministic-error}.
In the bowling scheme, duality enters only when we are ready to estimate the mean empirical width.

In our experience, this idea has been successful whenever we understand how to bound
the conic Gaussian width of the descent cone.  The main distinction is that the random
vector $\vct{\phi}$ may not share the rotational invariance of the standard Gaussian distribution.

\section{Example: Phase retrieval}
\label{sec:phase}

To demonstrate how the bowling scheme works, we consider the question of phase retrieval.
In this problem, we collect linear samples of an unknown signal, but we are only able to observe
their magnitudes.  To reconstruct the original signal, we must resolve the uncertainty about
the phases (or signs) of the measurements.
There is a natural convex program that can achieve this goal,
and the bowling scheme offers an easy way to analyze the number
of measurements that are required.

\subsection{Phase retrieval by convex optimization}

In the phase retrieval problem, we wish to recover a signal $\vct{x}^\natural \in \R^d$
from a family of measurements of the form
\begin{equation} \label{eqn:phase-nonlinear}
y_i = \abssqip{ \smash{\vct{x}^\natural} }{ \smash{\vct{\psi}_i} }
\quad\text{for $i = 1, 2, 3, \dots, m$.}
\end{equation}
The sampling ensemble $\vct{\psi}_1, \dots, \vct{\psi}_m$ consists of known vectors in $\R^d$.
For clarity of presentation, we do not consider the case where the samples are noisy or complex-valued.

Although the samples do not initially appear linear, we can apply a lifting method
proposed by Balan et al.~\cite{BBCE09:Painless-Reconstruction}.  Observe that
\begin{equation*}
\abssqip{ \vct{x} }{ \smash{\vct{\psi}} }
	= \vct{\psi}^\transp \vct{x} \cdot \vct{x}^\transp \vct{\psi}
	= \trace\big( \vct{x}\vct{x}^\transp \cdot \vct{\psi} \vct{\psi}^\transp \big).
\end{equation*}
In view of this expression, it is appropriate to introduce the rank-one positive-semidefinite matrices
\begin{equation} \label{eqn:lifting}
\mtx{X}^\natural = (\vct{x}^\natural) (\vct{x}^\natural)^\transp \in \R^{d \times d}
\quad\text{and}\quad
\mtx{\Psi}_i = \vct{\psi}_i \vct{\psi}_i^\transp \in \R^{d\times d}
\quad\text{for $i = 1, 2, 3, \dots, m$.}
\end{equation}
Then we can express the samples $y_i$ as \emph{linear} functions of the matrix $\mtx{X}^\natural$:
\begin{equation} \label{eqn:phase-data}
y_i = \trace\big( \mtx{X}^\natural \cdot \mtx{\Psi}_i \big)
\quad\text{for $i = 1, 2, 3, \dots, m$.}
\end{equation}
The expression~\eqref{eqn:phase-data} coincides with
the measurement model~\eqref{eqn:data-acquisition}
we have been considering.

We can use convex optimization to reconstruct the unknown matrix $\mtx{X}^\natural$.
It is natural to minimize the Schatten 1-norm to promote low rank,
but we also want to enforce the fact that $\mtx{X}^\natural$ is positive semidefinite~\cite{Faz02:Matrix-Rank}.
To that end, we consider the convex program
\begin{equation} \label{eqn:phase-retrieval}
\minimize{\mtx{X} \in \R^{d \times d}} \trace(\mtx{X})
\subjto
\mtx{X} \psdge \mtx{0}
\quad\text{and}\quad
y_i = \trace\big( \mtx{X} \mtx{\Psi}_i \big)
\quad\text{for each $i = 1, 2, 3, \dots, m$.}
\end{equation}
This formulation involves the lifted variables~\eqref{eqn:lifting}.
We say that the optimization problem~\eqref{eqn:phase-retrieval} \term{recovers} $\vct{x}^\natural$
if the matrix $\mtx{X}^\natural$ is the unique minimizer.  Indeed, in this
case, we can reconstruct the original signal by factorizing the solution
to the optimization problem.

\begin{remark}[Citation for convex phase retrieval]
The formulation~\eqref{eqn:phase-retrieval}
was developed by a working group at the meeting
``Frames for the finite world: Sampling, coding and quantization,''
which took place at the American Institute of Mathematics in Palo Alto in August 2008.
Most of the recent literature attributes
this idea incorrectly.
\end{remark}

\subsection{Phase retrieval from Gaussian measurements}

Recently, researchers have started to consider phase retrieval
problems with random data; see~\cite{CSV13:Phase-Lift} for example.
In the simplest instance,
we choose each sampling vector $\vct{\psi}_i$
independently from the standard normal distribution on $\R^d$:
$$
\vct{\psi}_i \sim \normal(\vct{0}, \Id_d).
$$
Then each sampling matrix $\mtx{\Psi}_i = \vct{\psi}_i \vct{\psi}_i^\transp$
follows a Wishart distribution.  These random matrices do not have
subgaussian marginals, so we cannot apply Corollary~\ref{cor:recovery-subgauss}
to study the performance of the optimization problem~\eqref{eqn:phase-retrieval}.
Nevertheless, we can make short work of the analysis by using the bowling scheme.

\begin{theorem}[Phase retrieval from Gaussian measurements] \label{thm:phase-retrieval}
Let $\vct{x}^\natural$ be a signal in $\R^d$.
Let $\vct{\psi}_i \sim \normal(\vct{0}, \Id_d)$
be independent standard Gaussian vectors,
and consider random measurements $y_i = \abssqip{ \smash{\vct{x}^\natural} }{ \smash{ \vct{\psi}_i } }$
for $i = 1, 2, 3, \dots, m$.
Assuming that $m \geq Cd$,
the convex phase retrieval problem~\eqref{eqn:phase-retrieval}
recovers $\vct{x}^\natural$ with probability at least $1 - \econst^{-cm}$.
The numbers $c$ and $C$ are positive absolute constants.
\end{theorem}

The sampling complexity $m \geq Cd$ established in Theorem~\ref{thm:phase-retrieval}
is qualitatively optimal.  Indeed, a dimension-counting argument shows that we need at least $m \geq d$
nonadaptive linear measurements to reconstruct a general vector in $\vct{x}^\natural \in \R^d$.

\begin{remark}[Extensions]
There are a number of obvious improvements to Theorem~\ref{thm:phase-retrieval}
that follow with a little more effort.
For example, it is clear that the convex phase retrieval method is stable.
The exceedingly high success probability also allows us to establish
uniform results for all $d$-dimensional vectors by means of net arguments and union bounds.
Furthermore, the Gaussian assumption is inessential; it is possible to establish
similar theorems for other sampling distributions.
We leave these refinements for the avid reader.
\end{remark}

\subsection{Proof of Theorem~\ref{thm:phase-retrieval}: Setup}

Let us rewrite the optimization problem~\eqref{eqn:phase-retrieval} in a form
that is more conducive to our methods of analysis.  First, introduce the
inner product space $\R^{d \times d}_{\rm sym}$ of $d \times d$
symmetric matrices, equipped with the trace inner product
$\ip{\mtx{B}}{\mtx{A}} := \trace(\mtx{AB})$
and the Frobenius norm $\fnorm{\cdot}$.  Define the linear operator
$$
\mtx{\Phi} : \R^{d\times d}_{\rm sym} \to \R^m
\quad\text{where}\quad
[ \mtx{\Phi}(\mtx{X}) ]_i = \ip{ \mtx{X} }{ \mtx{\Psi}_i }
\quad\text{for $i = 1, 2, 3, \dots, m$.}
$$
Collect the measurements into a vector $\vct{y} = (y_1, \dots, y_m)^\transp \in \R^m$,
and observe that $\vct{y} = \mtx{\Phi}(\mtx{X}^\natural)$ because of the expression~\eqref{eqn:phase-data}.
Next, define the convex indicator function of the positive-semidefinite cone:
$$
\iota : \R^{d\times d}_{\rm sym} \to \overline{\R}
\quad\text{where}\quad
\iota(\mtx{X})
	= \begin{cases} 0, & \text{$\mtx{X}$ is positive semidefinite} \\
	+\infty, & \text{otherwise}.
	\end{cases}
$$
Introduce the convex regularizer
$$
f : \R^{d \times d}_{\rm sym} \to \overline{\R}
\quad\text{where}\quad
f(\mtx{X}) = \trace(\mtx{X}) + \iota(\mtx{X}).
$$
With this notation, we can write~\eqref{eqn:phase-retrieval} in the form
\begin{equation} \label{eqn:cpt-phase-retrieval}
\minimize{\mtx{X}\in\R^{d\times d}_{\rm sym}} f(\mtx{X})
\subjto \vct{y} = \mtx{\Phi}(\mtx{X}).
\end{equation}
The formulation~\eqref{eqn:cpt-phase-retrieval} matches our core problem~\eqref{eqn:convex-recovery}
with the error vector $\vct{e} = \vct{0}$ and error tolerance $\eta = 0$.

Proposition~\ref{prop:deterministic-error} demonstrates that
$\mtx{X}^\natural$ is the unique solution of~\eqref{eqn:cpt-phase-retrieval}
whenever
$$
\lambda_{\min}\big(\mtx{\Phi};\ \Desc(f, \mtx{X}^\natural) \big)
	> 0.
$$
We must determine how many measurements $m$
suffice for this event to hold with high probability.

\subsection{Step 1: The nonnegative empirical process bound}

Define the set
$$
E := \big\{ \mtx{U} \in \Desc(f, \mtx{X}^\natural) : \fnorm{\mtx{U}} = 1 \big\}
	\subset \R^{d \times d}_{\rm sym}.
$$
Proposition~\ref{prop:km} demonstrates that
\begin{equation} \label{eqn:phase-km}
\lambda_{\min}\big(\mtx{\Phi};\ \Desc(f, \mtx{X}^\natural) \big)
	= \inf_{\mtx{U} \in E}\ \left( \sum_{i=1}^m \abssqip{\mtx{U}}{\mtx{\Psi}_i} \right)^{1/2}
	\geq \xi \sqrt{m} \, Q_{2\xi}(E) - 2 \, W_m(E) - \xi t
\end{equation}
with probability at least $1 - \econst^{-t^2/2}$.
In this setting, the marginal tail function is defined as
$$
Q_{2\xi}(E) := \inf_{\mtx{U} \in E} \ \Prob\big\{ \absip{ \mtx{U} }{ \mtx{\Psi}_1 } \geq 2 \xi \big\}.
$$
The mean empirical width is defined as
$$
W_m(E) := \Expect{} \sup_{\mtx{U} \in E} \ \ip{ \mtx{U} }{ \mtx{H} }
	\quad\text{where}\quad
	\mtx{H} := \frac{1}{\sqrt{m}} \sum_{i=1}^m \eps_i \mtx{\Psi}_i.
$$
Here, $\{\eps_i\}$ is an independent family of Rademacher random variables,
independent from everything else.

\subsection{Step 2: The marginal tail function}

We can use the Paley--Zygmund inequality to show that
\begin{equation} \label{eqn:phase-tail}
Q_1(E) = \inf_{\mtx{U} \in E} \ \Prob\big\{ \absip{ \mtx{U} }{ \mtx{\Psi}_1 } \geq 1 \big\}
	\geq c_0.
\end{equation}
We have implicitly chosen $\xi = \half$, and $c_0$ is a positive absolute constant.

\subsubsection{The tail bound}

To perform this estimate, we apply the Paley--Zygmund inequality in the form
$$
\Prob\bigg\{ \abssqip{\mtx{U}}{\mtx{\Psi}_1} \geq \tfrac{1}{2} \big(\Expect{} \abssqip{\mtx{U}}{\mtx{\Psi}_1}\big) \bigg\}
	\geq \frac{1}{4} \cdot \frac{ \big( \Expect{} \abssqip{ \mtx{U} }{ \mtx{\Psi}_1 } \big)^2 }
	{\Expect{} \absip{\mtx{U}}{ \mtx{\Psi}_1 }^4}.
$$
The easiest way to treat the expectation in the denominator is to invoke
Gaussian hypercontractivity~\cite[Sec.~3.2]{LT91:Probability-Banach}.  Indeed,
$$
\big( \Expect{} \absip{ \mtx{U} }{ \mtx{\Psi}_1 }^4 \big)^{1/4}
	\leq C_0 \, \big( \Expect{} \abssqip{ \mtx{U} }{ \mtx{\Psi}_1 } \big)^{1/2}
$$
because $\ip{ \mtx{U} }{ \mtx{\Psi}_1 }$
is a second-order polynomial in the entries of $\vct{\psi}_1$.
Combine the last two displays to obtain
$$
\Prob\bigg\{ \abssqip{\mtx{U}}{\mtx{\Psi}_1} \geq \tfrac{1}{2} \big(\Expect{} \abssqip{\mtx{U}}{\mtx{\Psi}_1}\big) \bigg\}
	\geq\frac{1}{4 \cdot C_0^4} = c_0.
$$
We can bound the remaining expectation by means of an explicit calculation.
Assuming that $\mtx{U} \in E$,
\begin{align*}
\Expect{} \abssqip{ \mtx{U} }{ \mtx{\Psi}_1 }
	= 3 \sum_{i=1}^m  \abssq{u_{ii}} + 2 \, \sum_{i,j=1}^m \abssq{\smash{u_{ij}}} + \abssq{ \sum_{i=1}^m u_{ii} }
	\geq 2.
\end{align*}
We have used the fact that $\mtx{U}$ is a symmetric matrix with unit Frobenius norm.  In conclusion,
$$
\Prob\left\{ \abssqip{\mtx{U}}{\mtx{\Psi}_1} \geq 1 \right\} \geq c_0
\quad\text{for each $\mtx{U} \in E$.}
$$
This inequality implies~\eqref{eqn:phase-tail}.

\subsection{Step 3$'$: The mean empirical width of the descent cone}

We can apply Proposition~\ref{prop:emp-width-descent} to demonstrate that the mean empirical width satisfies
\begin{equation} \label{eqn:phase-width}
W_m(E) \leq C_1\sqrt{d}
\quad\text{for $m \geq C_2d$.}
\end{equation}
The numbers $C_1$ and $C_2$ are positive, absolute constants.

\subsubsection{The width bound}

The bound holds trivially when $\mtx{X}^\natural = \mtx{0}$, so we may assume that the unknown matrix is nonzero.
Select a coordinate system where
$$
\mtx{X}^\natural = \begin{bmatrix} a & \vct{0}^\transp \\ \vct{0} & \mtx{0} \end{bmatrix}
	\in \R^{d\times d}_{\rm sym}
\quad\text{where $a > 0$.}
$$
Recall that the matrix $\mtx{H} = m^{-1/2} \sum_{i=1}^m \eps_i \mtx{\Psi}_i$,
where $\mtx{\Psi}_i = \vct{\psi}_i \vct{\psi}_i^\transp$ and
$\vct{\psi}_i \sim \normal(\vct{0}, \Id_d)$.
Partition $\mtx{H}$ conformally with $\mtx{X}^\natural$:
$$
\mtx{H} = \begin{bmatrix} h_{11} & \vct{h}_{21}^\transp \\ \vct{h}_{21} & \mtx{H}_{22} \end{bmatrix}.
$$
Define the random parameter $\tau = \lambda_{\max}( \mtx{H}_{22} )$, where $\lambda_{\max}$ denotes
the maximum eigenvalue of a symmetric matrix. Proposition~\ref{prop:emp-width-descent} delivers the width bound
\begin{equation} \label{eqn:phase-dual-width}
W_{m}(E) = \Expect{} \sup_{\mtx{U} \in E} \ \ip{\mtx{U}}{\mtx{H}}
	\leq \left( \Expect{} \dist_{\rm F}^2\big( \mtx{H}, \ \tau \cdot \partial f(\mtx{X}^\natural) \big) \right)^{1/2} .
\end{equation}
Using standard calculus rules for subdifferentials~\cite[Chap.~23]{Roc70:Convex-Analysis},
we determine that
$$
\partial f(\mtx{X}^\natural) = \left\{ \begin{bmatrix} 1 & \vct{0}^\transp \\ \vct{0} & \mtx{Y} \end{bmatrix}
	\in \R^{d\times d}_{\rm sym} : \lambda_{\max}(\mtx{Y}) \leq 1 \right\}.
$$
Next,
\begin{align} \label{eqn:phase-dist}
\Expect{} \dist_{\rm F}^2 \big( \mtx{H}, \ \partial f(\mtx{X}^\natural) \big)
	= \Expect{} (h_{11} - \tau)^2 + 2\, \Expect{} \normsq{ \vct{h}_{21} }
	+ \Expect{} \inf_{\lambda_{\max}(\mtx{S}) \leq 1} \ \fnormsq{\mtx{H}_{22} - \tau \cdot \mtx{Y}}.
\end{align}
By construction, the third term on the right-hand side of~\eqref{eqn:phase-dist} is zero.
By direct calculation, the second term on the right-hand side of~\eqref{eqn:phase-dist}
satisfies
\begin{equation} \label{eqn:phase-term-2}
\Expect{} \normsq{ \vct{h}_{21} } = d - 1.
\end{equation}
Finally, we turn to the first term on the right-hand side of~\eqref{eqn:phase-dist}.
Relatively crude bounds suffice here.
By interlacing of eigenvalues,
$$
\tau = \lambda_{\max}( \mtx{H}_{22} )
	\leq \lambda_{\max}( \mtx{H} )
	= \frac{1}{\sqrt{m}} \lambda_{\max}\left( \sum_{i=1}^m \eps_i \vct{\psi}_i \vct{\psi}_i^\transp \right).
$$
Standard net arguments, such as those in~\cite[Sec.~5.4.1]{Ver11:Introduction-Nonasymptotic},
demonstrate that
$$
\Prob\big\{ \lambda_{\max}(\mtx{H}) \geq C_3\sqrt{d} \big\}
	\leq \econst^{-c_1d},
	\quad\text{provided that $m \geq C_2 d$.}
$$
Together, the last two displays imply that $\Expect{} \tau^2 \leq C_4 d$.  Therefore,
\begin{equation} \label{eqn:phase-term-1}
\Expect{} (h_{11} - \tau)^2
	\leq C_5 d.
\end{equation}
Introducing~\eqref{eqn:phase-dist},~\eqref{eqn:phase-term-2}, and~\eqref{eqn:phase-term-1}
into~\eqref{eqn:phase-dual-width}, we arrive at the required bound~\eqref{eqn:phase-width}.

\begin{remark}[Other sampling distributions]
The only challenging part of the calculation is the bound on $\lambda_{\max}(\mtx{H})$.
For more general sampling distributions, we can easily obtain the required estimate from
the matrix moment inequality~\cite[Thm.~A.1]{CGT12:Masked-Sample}.
\end{remark}

\subsection{Combining the bounds}

Assume that $m \geq C_2 d$.
Combine the estimates~\eqref{eqn:phase-km},~\eqref{eqn:phase-tail}, and~\eqref{eqn:phase-width}
to reach
$$
\lambda_{\min}\big(\mtx{\Phi};\ \Desc(f, \mtx{X}^\natural) \big)
	\geq c_2 \sqrt{m} - C_6 \sqrt{d} - \half t
$$
with probability at least $1 - \econst^{-t^2/2}$.
Choosing $t = c_3 \sqrt{m}$, we find that the minimum
conic singular value is positive with probability
at least $1 - \econst^{-c_4 m}$.
In this event, Proposition~\ref{prop:deterministic-error}
implies that $\mtx{X}^\natural$ is the unique solution to the phase retrieval
problem~\eqref{eqn:phase-retrieval}.
This observation completes the proof of Theorem~\ref{thm:phase-retrieval}.

\section*{Acknowledgments}

JAT gratefully acknowledges support from ONR award N00014-11-1002, AFOSR award FA9550-09-1-0643, and a Sloan Research Fellowship.  Thanks are also due to the Moore Foundation.

\bibliographystyle{myalpha}

\begin{thebibliography}{TWD{\etalchar{+}}06}

\bibitem[ALMT14]{ALMT14:Living-Edge-II}
D.~Amelunxen, M.~Lotz, M.~B. McCoy, and J.~A. Tropp.
\newblock Living on the edge: Phase transitions in convex programs with random
  data.
\newblock {\em Inform. Inference}, 3(3):224--294, 2014.
\newblock Available at \url{http://arXiv.org/abs/1303.6672}.

\bibitem[BBCE09]{BBCE09:Painless-Reconstruction}
R.~Balan, B.~G. Bodmann, P.~G. Casazza, and D.~Edidin.
\newblock Painless reconstruction from magnitudes of frame coefficients.
\newblock {\em J. Fourier Anal. Appl.}, 15(4):488--501, 2009.

\bibitem[BLM13]{BLM13:Concentration-Inequalities}
S.~Boucheron, G.~Lugosi, and P.~Massart.
\newblock {\em Concentration inequalities: A nonasymptotic theory of
  independence}.
\newblock Oxford University Press, Inc., 2013.

\bibitem[CGT12]{CGT12:Masked-Sample}
R.~Y. Chen, A.~Gittens, and J.~A. Tropp.
\newblock The masked sample covariance estimator: {A}n analysis via the matrix
  {L}aplace transform method.
\newblock {\em Information and Inference}, 1:2--20, 2012.

\bibitem[CLR14]{TLR14:Geometrizing-Local}
T.~T. Cai, T.~Liang, and A.~Rakhlin.
\newblock Geometrizing local rates of convergence for linear inverse problems.
\newblock Available at \url{http://arXiv.org/abs/1404.4408}, Apr. 2014.

\bibitem[CRPW12]{CRPW12:Convex-Geometry}
V.~Chandrasekaran, B.~Recht, P.~A. Parrilo, and A.~S. Willsky.
\newblock The convex geometry of linear inverse problems.
\newblock {\em Found. Comput. Math.}, 12(6):805--849, 2012.

\bibitem[CSV13]{CSV13:Phase-Lift}
E.~J. Cand{\`e}s, T.~Strohmer, and V.~Voroninski.
\newblock Phase{L}ift: exact and stable signal recovery from magnitude
  measurements via convex programming.
\newblock {\em Comm. Pure Appl. Math.}, 66(8):1241--1274, 2013.

\bibitem[DS01]{DS01:Local-Operator}
K.~R. Davidson and S.~J. Szarek.
\newblock Local operator theory, random matrices and {B}anach spaces.
\newblock In {\em Handbook of the geometry of {B}anach spaces, {V}ol. {I}},
  pages 317--366. North-Holland, Amsterdam, 2001.

\bibitem[Faz02]{Faz02:Matrix-Rank}
M.~Fazel.
\newblock {\em Matrix rank minimization with applications}.
\newblock PhD thesis, Stanford University, 2002.

\bibitem[FM14]{FoyMac:13}
R.~Foygel and L.~Mackey.
\newblock Corrupted sensing: Novel guarantees for separating structured
  signals.
\newblock {\em Trans. Inform. Theory}, 60(2):1223--1247, Feb. 2014.

\bibitem[Gor85]{Gor85:Some-Inequalities}
Y.~Gordon.
\newblock Some inequalities for {G}aussian processes and applications.
\newblock {\em Israel J. Math.}, 50(4):265--289, 1985.

\bibitem[Gor88]{Gor88:Milmans-Inequality}
Y.~Gordon.
\newblock On {M}ilman's inequality and random subspaces which escape through a
  mesh in {${\bf R}\sp n$}.
\newblock In {\em Geometric aspects of functional analysis (1986/87)}, volume
  1317 of {\em Lecture Notes in Math.}, pages 84--106. Springer, Berlin, 1988.

\bibitem[Gro11]{Gro11:Recovering-Low-Rank}
D.~Gross.
\newblock Recovering low-rank matrices from few coefficients in any basis.
\newblock {\em IEEE Trans. Inform. Theory}, 57(3):1548--1566, Mar. 2011.

\bibitem[KM13]{KM13:Bounding-Smallest}
V.~Koltchinskii and S.~Mendelson.
\newblock Bounding the smallest singular value of a random matrix without
  concentration.
\newblock Available at \url{http://arXiv.org/abs/1312.3580}, Dec. 2013.

\bibitem[LM14]{LM14:Compressed-Sensing}
G.~Lecu{\'e} and S.~Mendelson.
\newblock Compressed sensing under weak moment assumptions.
\newblock Available at \url{http://arXiv.org/abs/1401.2188}, Jan. 2014.

\bibitem[LO94]{LO94:Best-Constant}
R.~Lata{\l}a and K.~Oleszkiewicz.
\newblock On the best constant in the {K}hinchin-{K}ahane inequality.
\newblock {\em Studia Math.}, 109(1):101--104, 1994.

\bibitem[LT91]{LT91:Probability-Banach}
M.~Ledoux and M.~Talagrand.
\newblock {\em Probability in Banach Spaces: Isoperimetry and Processes}.
\newblock Springer, Berlin, 1991.

\bibitem[Men10]{Men10:Empirical-Processes}
S.~Mendelson.
\newblock Empirical processes with a bounded {$\psi\sb 1$} diameter.
\newblock {\em Geom. Funct. Anal.}, 20(4):988--1027, 2010.

\bibitem[Men13]{Men13:Remark-Diameter}
S.~Mendelson.
\newblock A remark on the diameter of random sections of convex bodies.
\newblock Available at \url{http://arXiv.org/abs/1312.3608}, Dec. 2013.

\bibitem[Men14a]{Men14:Learning-Concentration}
S.~Mendelson.
\newblock Learning without concentration.
\newblock To appear, \textsl{J. Assoc. Comput. Mach.} Available at
  \url{http://arXiv.org/abs/1401.0304}, Jan. 2014.

\bibitem[Men14b]{Men14:Learning-Concentration-II}
S.~Mendelson.
\newblock Learning without concentration for general loss functions.
\newblock Available at \url{http://arXiv.org/abs/1410.3192}, Oct. 2014.

\bibitem[MPTJ07]{MPT07:Reconstruction-Subgaussian}
S.~Mendelson, A.~Pajor, and N.~Tomczak-Jaegermann.
\newblock Reconstruction and subgaussian operators in asymptotic geometric
  analysis.
\newblock {\em Geom. Funct. Anal.}, 17(4):1248--1282, 2007.

\bibitem[MS86]{MS86:Asymptotic-Theory}
V.~Milman and G.~Schechtman.
\newblock {\em Asymptotic theory of finite-dimensional normed linear spaces}.
\newblock Number 1200 in LNM. Springer, 1986.

\bibitem[OH10]{OH:10}
S.~Oymak and B.~Hassibi.
\newblock New null space results and recovery thresholds for matrix rank
  minimization.
\newblock Partial results presented at ISIT 2011. Available at
  \url{http://arXiv.org/abs/1011.6326}, 2010.

\bibitem[OH13]{OymHas:13}
S.~Oymak and B.~Hassibi.
\newblock Sharp {MSE} bounds for proximal denoising.
\newblock Partial results presented at Allerton 2012. Available at
  \url{http://arxiv.org/abs/1305.2714}, March 2013.

\bibitem[OTH13]{OTH13:Simple-Bounds}
S.~Oymak, C.~Thrampoulides, and B.~Hassibi.
\newblock Simple bounds for noisy linear inverse problems with exact side
  information.
\newblock Available at \url{http://arXiv.org/abs/1312.0641}, Dec. 2013.

\bibitem[Pis89]{Pis89:Volume-Convex}
G.~Pisier.
\newblock {\em The Volume of Convex Bodies and Banach Space Geometry}.
\newblock Cambridge Univ. Press, 1989.

\bibitem[Roc70]{Roc70:Convex-Analysis}
R.~T. Rockafellar.
\newblock {\em Convex Analysis}.
\newblock Princeton Univ. Press, 1970.

\bibitem[RV08]{RV08:Sparse-Reconstruction}
M.~Rudelson and R.~Vershynin.
\newblock On sparse reconstruction from {F}ourier and {G}aussian measurements.
\newblock {\em Comm. Pure Appl. Math.}, 61(8):1025--1045, 2008.

\bibitem[Sio58]{Sio58:General-Minimax}
M.~Sion.
\newblock On general minimax theorems.
\newblock {\em Pacific J. Math}, 8:171--176, 1958.

\bibitem[Sto09]{stojnic10}
M.~Stojnic.
\newblock Various thresholds for $\ell_1$-optimization in compressed sensing.
\newblock Available at \url{http://arXiv.org/abs/0907.3666}, 2009.

\bibitem[Tal05]{Tal00:Generic-Chaining}
M.~Talagrand.
\newblock {\em The generic chaining}.
\newblock Springer Monographs in Mathematics. Springer-Verlag, Berlin, 2005.
\newblock Upper and lower bounds of stochastic processes.

\bibitem[TOH14]{TOH14:Simple-Error}
C.~Thrampoulides, S.~Oymak, and B.~Hassibi.
\newblock Simple error bounds for regularized noisy linear inverse problems.
\newblock Appeared at ISIT 2014. Available at
  \url{http://arXiv.org/abs/1401.6578}, Jan. 2014.

\bibitem[TWD{\etalchar{+}}06]{TWDBB06:Random-Filters}
J.~Tropp, M.~Wakin, M.~Duarte, D.~Baron, and R.~Baraniuk.
\newblock Random filters for compressive sampling and reconstruction.
\newblock In {\em Acoustics, Speech and Signal Processing, 2006. ICASSP 2006
  Proceedings. 2006 IEEE International Conference on}, volume~3, pages
  III--III, May 2006.

\bibitem[vdVW96]{VW96:Weak-Convergence}
A.~W. van~der Vaart and J.~A. Wellner.
\newblock {\em Weak convergence and empirical processes}.
\newblock Springer Series in Statistics. Springer-Verlag, New York, 1996.
\newblock With applications to statistics.

\bibitem[Ver12]{Ver11:Introduction-Nonasymptotic}
R.~Vershynin.
\newblock Introduction to the non-asymptotic analysis of random matrices.
\newblock In {\em Compressed sensing}, pages 210--268. Cambridge Univ. Press,
  Cambridge, 2012.

\bibitem[Wat92]{Wat:92}
G.~A. Watson.
\newblock Characterization of the subdifferential of some matrix norms.
\newblock {\em Linear Algebra Appl.}, 170:33--45, 1992.

\end{thebibliography}
\newcommand{\etalchar}[1]{$^{#1}$}

\end{document}